\documentclass[dvips,preprint,authoryear,12pt]{imsart}
\usepackage{threeparttable}

\RequirePackage{natbib,amssymb,verbatim,undertilde}
\RequirePackage[OT1]{fontenc}
\RequirePackage{amsthm,amsmath,natbib,mathrsfs}
\RequirePackage[colorlinks,citecolor=blue,urlcolor=blue]{hyperref}
\RequirePackage{hypernat}
\usepackage{graphicx}
\startlocaldefs
\theoremstyle{plain}

\newtheorem{thm}{Theorem}

\newtheorem{cor}{Corollary}
\newtheorem{lemma}{Lemma}
\newtheorem*{lemma*}{Lemma}
\newtheorem{defn}{Definition}
\newtheorem{prop}{Proposition}
\theoremstyle{remark}
\newtheorem{rk}{Remark}

\newtheorem{ex}{Example}
\endlocaldefs

\def\MCP{\mathrm{MC+}}
\def\lasso{\mathrm{lasso}}
\def\Cov{\mathrm{Cov}}

\def\a{\alpha}

\def\m{\mathbf{m}}

\def\z{\mathbf{z}}

\def\col{\mathrm{col}}

\def\tr{\mathrm{tr}}

\def\1{\mathbf{1}}

\def\l{\lambda}
\def\bth{\boldsymbol{\theta}}

\def\e{\epsilon}
\def\ee{\boldsymbol{\epsilon}}
\def\b{\beta}
\def\T{\mathbf{T}}
\def\bb{\boldsymbol{\beta}}
\def\d{\delta}

\def\Var{\mathrm{Var}}

\def\g{\gamma}

\def\x{\mathbf{x}}

\def\s{\sigma}
\def\t{\mathbf{t}}
\def\u{\mathbf{u}}
\def\y{\mathbf{y}}

\def\1{\mathbf{1}}

\def\t{\mathbf{t}}
\def\R{\mathbb{R}}

\def\X{\mathbf{X}}
\def\S{\mathit{\Sigma}}

\def\vv{\mathbf{v}}
\def\iidsim{\stackrel{\scriptscriptstyle\mathrm{iid}}{\sim}}
\begin{document}

\bibliographystyle{ims}
\begin{frontmatter}
\title{Residual variance and the signal-to-noise ratio in
  high-dimensional linear models  \protect}

\runtitle{Residual variance and the signal-to-noise ratio}

\begin{aug}
\author{\fnms{Lee H.} \snm{Dicker}\thanksref{t1}
\ead[label=e1]{ldicker@stat.rutgers.edu}}

\thankstext{t1}{Supported by NSF Grant DMS-1208785}

\runauthor{L.H. Dicker}

\affiliation{Rutgers University}

\address{Department of Statistics and Biostatistics \\ Rutgers University \\ 501 Hill Center, 
 110 Frelinghuysen Road \\ Piscataway, NJ 08854 \\
\printead{e1}}
\end{aug}



\begin{abstract}
Residual variance and the signal-to-noise ratio are important
quantities in many statistical models and model fitting procedures.
They play an important role in regression diagnostics, in determining
the performance limits in estimation and prediction problems, and in
shrinkage parameter selection in many popular regularized regression
methods for high-dimensional data analysis.  We propose new estimators for
the residual variance, the $\ell^2$-signal strength, and the
signal-to-noise ratio that are consistent and asymptotically normal in
high-dimensional linear models with Gaussian predictors and errors,
where the number of predictors $d$ is proportional to the number of
observations $n$.  Existing results on residual variance estimation in
high-dimensional linear models depend on sparsity in the underlying
signal.  Our results require no sparsity assumptions and imply that
the residual variance may be consistently estimated even when $d > n$
and the underlying signal itself is non-estimable.  Basic numerical
work suggests that some of the distributional assumptions made for our
theoretical results may be relaxed. 
\end{abstract}

\end{frontmatter}
{\allowdisplaybreaks
\section{Introduction} 

Consider the linear model 
\begin{equation}\label{lm}
y_i = \x_i^T\bb + \e_i, \ \ i = 1,...,n,
\end{equation}
where $y_1,...,y_n \in \R$ and $\x_1 =
(x_{11},...,x_{1d})^T,...,\x_n = (x_{n1},...,x_{nd})^T\in \R^d$ are
observed outcomes and $d$-dimensional predictors, respectively,
$\e_1,...,\e_n \in \R$ are unobserved iid errors with $E(\e_i) = 0$
and $\Var(\e_i) = \s^2 > 0$, and $\bb = (\b_1,...,\b_d)^T \in \R^d$ is
an unknown $d$-dimensional parameter. To simplify notation, let $\y =
(y_1,...,y_n)^T \in \R^n$ denote the $n$-dimensional vector of
outcomes and $X = (\x_1,...,\x_n)^T$ denote the $n \times d$ matrix of
predictors.  Also let $\ee = (\e_1,...,\e_n)^T$.  Then (\ref{lm}) may
be re-expressed as
\[
\y = X\bb + \ee.  
\]
In this paper, we focus on the case where the predictors $\x_i$ are
random.  More specifically, we assume that $\x_1,...,\x_n$ are iid
random vectors with mean 0 and $d \times d$ positive definite covariance matrix
$\Sigma$ (many of the results in this paper are applicable if $E(\x_i)
\neq 0$ upon centering the data; however, this is not pursued further
here).  

Let $\tau^2 = \bb^T\S\bb = ||\S^{1/2}\bb||^2$, where $||\cdot||$
denotes the $\ell^2$-norm.  Then $\tau^2$ is a measure of the overall ($\ell^2$-)
signal strength.  The residual variance $\s^2 = \Var(\e_i) =
\Var\{E(y_i|\x_i)\}$ and the signal strength $\tau^2$ are important quantities in
many problems in statistics.  For example,   in estimation and
prediction problems, $\s^2$ typically determines
the scale of an estimator's risk under quadratic loss.   More broadly,
$\s^2$,
$\tau^2$, and associated quantities, such as the
signal-to-noise ratio $\tau^2/\s^2$, all play a key role in regression
diagnostics.  Thus, reliable estimators
of $\s^2$ and $\tau^2$ are desirable.  

For invertible $X^TX$, let $\hat{\bb}_{ols} =
(X^TX)^{-1}X^T\y$ be the ordinary least squares estimator for $\bb$.  If $n - d \to \infty$, then
\begin{equation}\label{s0}
\hat{\s}^2_0 = \frac{1}{n-d}||\y - X\hat{\bb}_{ols}||^2 =
\frac{1}{n-d}||\y||^2 - \frac{1}{n-d}\y^TX(X^TX)^{-1}X^T\y
\end{equation}
is a consistent estimator for $\s^2$ and, under fairly mild additional
conditions, is asymptotically normal.  Consistent estimators for
$\tau^2$ can also be constructed.  For instance, if $n - d \to
\infty$, it is
easily seen that 
\begin{equation}\label{tau0}
\hat{\tau}_0^2 = \frac{1}{n}||\y||^2 - \hat{\s}^2 =
-\frac{d}{n(n-d)}||\y||^2 + \frac{1}{n-d}\y^TX(X^TX)^{-1}X^T\y
\end{equation}
is a consistent estimator for $\tau^2$ under mild conditions.

It is more challenging to construct reliable estimators for $\s^2$ and
$\tau^2$ in high-dimensional linear models, where $d \geq n$.  Indeed,
if $d \geq n$,
then the estimator $\hat{\s}_0^2$ breaks
down; however, estimating $\s^2$ and $\tau^2$ remains important.  In
high-dimensional linear models with $d \geq n$, $\s^2$ plays
an important role in selecting
effective shrinkage parameters for many popular regularized regression
methods \citep{candes2007dantzig,
  bickel2009simultaneous,zhang2010nearly}.  The
signal-to-noise ratio $\tau^2/\s^2$ is also important for shrinkage
parameter selection, and it determines performance limits
in certain high-dimensional regression problems
\citep{dicker2012dense, dicker2012optimal}.  

In this paper, we propose new estimators for $\s^2$ and $\tau^2$ that are
consistent and asymptotically normal, with rate $n^{-1/2}$, in an asymptotic regime where
$d/n \to \rho \in [0,\infty)$ (whenever we write $d/n \to \rho$, it is
implicit that $n \to \infty$ as well).   We also show that these
estimators may be used to derive consistent and asymptotically normal
estimators for function sof $\s^2$ and $\tau^2$, like the
signal-to-noise ratio.  Previous work on estimating $\s^2$
in high-dimensional linear models where $d \geq n$ has been conducted by  \cite{sun2011scaled} and
\cite{fan2012variance}.  These authors assume that $\bb$ is sparse
(e.g. the $\ell^1$-norm or $\ell^0$-norm of $\bb$ is small) and their
results for estimating $\s^2$ are related to the fact that
$\bb$ itself is estimable under the specified sparsity assumptions.
Though Sun and Zhang's (2011) and Fan et al.'s (2012) results even apply in settings
where $d/n \to \infty$, their sparsity assumptions may be untenable in
certain instances and this can dramatically affect the performance of
their estimators.  In this paper, we  make no sparsity assumptions (however, $\s^2$ and $\tau^2$ are
required to be bounded) and we show that the proposed estimators for
$\s^2$ and $\tau^2$ perform well in situations where $d \geq n$ and $\bb$ is provably
non-estimable.  This is one of the main messages of the paper:  Though
some type of sparsity is required to consistently estimate $\bb$ in
high-dimensional linear models, sparsity in $\bb$ is {\em not}
required to  estimate $\s^2$  and $\tau^2$. 

\subsection{Distributional assumptions}

Though sparsity is not required in this paper, we do make
strong distributional assumptions about the data.  In particular, we
henceforth assume that
\begin{equation}\label{normality}
\e_1,...,\e_n \iidsim N(0,\s^2) \ \mbox{ and } \ \x_1,...,\x_n \iidsim
N(0,\S).
\end{equation}
While normality is used heavily throughout our analysis, we expect that key aspects of many of the results in this paper remain valid under weaker
distributional assumptions.  This is explored via simulation in
Section 4.     

Not surprisingly, the analysis in this paper is simplified
by the normality assumption (\ref{normality}).  To explain the
relevance of (\ref{normality}) in more detail, we first point out that
our primary consistency
results for the proposed
estimators of $\s^2$ and $\tau^2$ (Theorem 1
below) follow from exact calculations of
the estimators' mean and variance.   If
the normality assumption (\ref{normality}) is violated, then these
calculations are generally invalid; similar techniques may be
applicable, if other conditions hold, but exact finite sample
calculations are not likely to be possible and any corresponding
approximation may be more involved.    

The normality
assumption (\ref{normality}) also facilitates the use of a collection of
``soft-tools'' for random matrices developed by \cite{chatterjee2009fluctuations} 
to prove that the estimators proposed in this paper are asymptotically normal.
These tools are related to second order Poincar\'e
inequalities and Stein's method \citep{stein1986approximate}.
Asymptotic normality for the proposed estimators follows by bounding the total variation
distance to a normal random variable.  These bounds contain
information about how the variability of the proposed estimators may
depend $d$, $n$, $\S$, $\s^2$, and $\tau^2$.   This is easily
leveraged to obtain consistent and asymptotically normal estimators
for functions of $\s^2$ and $\tau^2$ (such as the
signal-to-noise ratio, $\tau^2/\s^2$; see Corollary 2 below), which is an important practical
objective.  Thus, one of the appealing aspects of the ``soft tools''
used in this paper is their flexibility.   On the other hand, paraphrasing
\cite{chatterjee2009fluctuations}, other existing methods for
asymptotic analysis in random matrix theory rely heavily on the
exact calculation of limits \citep{jonsson1982some, bai2004clt};
we suggest that this may be a more delicate endeavor in some
instances.   If the normality assumption (\ref{normality})
does not hold, then it is unclear if the soft tools used in this paper
are still applicable and, consequently, other techniques may be required.  Existing work in random
matrix theory suggests that this may be possible  (see, for example,
\citep{bai2007asymptotics, pan2008central, elkaroui2011geometric}); however, the computations are likely
more involved and the breadth of applicability of alternative techniques seems unclear.

\subsection{Correlation among predictors} 

Another challenging issue for estimating $\s^2$ and $\tau^2$ when $d >
n$ involves the covariance matrix $\Cov(\x_i) = \S$.  Our initial estimators for
$\s^2$ and $\tau^2$ are devised under the assumption that $\S$ is
known (equivalently, $\S = I$; see Section 2).  These
estimators are unbiased, consistent, and asymptotically normal.  We subsequently
propose modified estimators for $\s^2$ and $\tau^2$ in cases where
$\S$ is unknown, but (i) a norm-consistent estimator for $\S$ is
available, or (ii)  $\S$ and $\bb$ satisfy certain conditions described in Section 3.2.  If a norm-consistent estimator for $\S$ is available, then the
  proposed estimators for $\s^2$ and $\tau^2$ are consistent;  if, furthermore,
  $\S$ is estimated at rate $o(n^{-1/2})$, then the estimators are asymptotically normal.
 On the other hand, if $d/n \to \rho \in (0,\infty)$, then norm-consistent
 estimators for $\S$ are not generally available (though there are
 important examples where norm-consistent estimators for $\S$ can be
 found -- this is discussed in more detail in Section 3.1).  Thus, it is important to construct
 estimators for $\s^2$ and $\tau^2$ that perform reliably when $\S$ is completely
 unknown.  While it remains an open problem to find
 estimators for $\s^2$ and $\tau^2$ that are consistent for completely general $\S$, in
 Section 3.2 we propose estimators that are consistent and
 asymptotically normal, provided $\S$ and $\bb$ satisfy conditions
 that are closely related to other conditions that have appeared in
 the random matrix theory literature \citep{bai2007asymptotics,
   pan2008central}.  These conditions basically require that $\bb$ and
 $\S$ are {\em asymptotically free} in the sense of free probability
 (see, for example, \citep{speicher2003free} for a brief overview of free
 probability and random matrix theory). 

\subsection{Additional remarks}

The problems considered in this paper have at least a passing
resemblance to the Neyman-Scott problem \citep{neyman1948consistent,
  lancaster2000incidental}.  In a simplified version of this problem, observations $w_{ij} \sim N(\mu_i,\nu^2)$, $i =
1,...,n$, $j = 1,2$ are available, and the goal is to estimate
$\s^2$.  The means $\mu_i$ are nuisance parameters and,
without additional specification,
none of the $\mu_i$ are estimable, as $n \to \infty$.  Furthermore, the profile maximum
likelihood estimator for $\nu^2$, which is given by 
\[
\hat{\nu}^2_{MLE} = \frac{1}{4n} \sum_{i = 1}^n (w_{i1} - w_{i2})^2,
\]
is inconsistent; indeed, $\lim_{n \to \infty} \hat{\nu}^2_{MLE} = \nu^2/2$.  On the other
hand, the simple method of moments estimator $\hat{\nu}_{MOM}^2 =
2\hat{\nu}^2_{MLE}$ is consistent for $\nu^2$ and asymptotically normal.  

In linear models (\ref{lm}) with $d \geq n$, which are the main focus
of this paper, the parameter $\bb$ is typically non-estimable.
However, we show below that $\s^2$ may still be consistently estimated
in a variety of circumstances.  Moreover, as in the Neyman-Scott
problem, it is unclear how to proceed with likelihood inference.
Indeed, the MLE
\[
\hat{\s}^2_{MLE} = \left\{\begin{array}{cl}\frac{1}{n}||\y -
    X\hat{\bb}_{ols}||^2 & \mbox{if } d < n \\
0 & \mbox{if } d \geq n \end{array}\right.
\]
is degenerate when $d \geq n$ and it can even be troublesome when $d < n$:  if $d/n \to \rho \in
(0,1)$, then $\hat{\s}^2 \to (1 - \rho)\s^2 \neq \s^2$.
Furthermore, similar to the Neyman-Scott problem described in the
previous paragraph,
the basic estimator for $\s^2$ derived
in Section 2.1 is a method of moments estimators. 

In our view, the
major implication of the preceding discussion is that the ambiguities of likelihood inference which
arise in this problem contribute to difficulties in devising a
systematic approach to estimation and efficiency when studying $\s^2$,
$\tau^2$, and related quantities in high-dimensional linear models.
While the estimators proposed in this paper are shown to have
reasonable properties, further research into these broader issues may be
warranted.  

\subsection{Overview of the paper}

Section 2 is primarily devoted to the case where $\Cov(\x_i) = I$.  A
motivating discussion and the definition of the basic estimators for
$\s^2$ and $\tau^2$ may be found in Section 2.1.  Section 2.2 and
Section 2.3 address consistency and asymptotic normality for the basic
estimators, respectively.  The case where $\Cov(\x_i)  = \S$ is
unknown is addressed in
Section 3.  Section 3.1 is concerned with the case where a  norm-consistent estimator for $\S$ is available; Section 3.2 covers
the case where no such estimator may be found, but $\bb$ and $\S$
satisfy certain additional conditions.  The results of three
simulation studies are reported in Section 4.  Two of these studies
illustrate basic properties of the estimators proposed in this paper.
In the third study, we compare the performance of our estimators
for $\s^2$ to the performance of estimators for $\s^2$
proposed by \cite{sun2011scaled}.  Section 5 contains a concluding
discussion, where we briefly mention some potential alternatives to
the estimators proposed in this paper and issues related to efficiency.
Proofs may be found in the Appendix; some of the more extended calculations required for these
proofs are contained in the Supplemental Text (which may be found
after the Bibliography below).

\section{Independent predictors: $\S = I$}

Throughout the discussion in this section, we assume that $\S = I$.  
All of the calculations in Section 2.1-2.2 require $\S = I$.  However,
the main result of Section 2.3 (Theorem 3, on asymptotic normality)
holds for arbitrary positive definite $\S$.  Notice that if $\S \neq
I$, but $\S$ is known, then one easily reduces to the case where $\S = I$ be replacing $X$ with $X\S^{-1/2}$. 

\subsection{Motivation and the basic estimators}

For illustrative purposes, suppose for the moment that $d < n$. The
estimator $\hat{\s}_0^2$, defined in (\ref{s0}), may be interpreted as the projection of
$\y$ onto $\col(X)^{\perp} \subseteq \R^n$, the orthogonal
complement of the column space of $X$.  This well-known interpretation
highlights one of the obstacles to estimating
$\s^2$ in linear models with more predictors than observations:  If $d
\geq n$, then $\col(X) = \R^n$; thus, $\col(X)^{\perp} = \{0\}$ and
any projection onto $\col(X)^{\perp}$ is trivial.  An alternative
interpretation of $\hat{\s}_0^2$ suggests methods for estimating
$\s^2$ and $\tau^2$ in high-dimensional linear models.  

Consider the
linear combination of $n^{-1}||\y||^2$ and
$n^{-1}\y^TX(X^TX)^{-1}X^T\y$, 
\[
L_0(a_1,a_2) = a_1\frac{1}{n}||\y||^2 + a_2\frac{1}{n}\y^TX(X^TX)^{-1}X^T\y
\]
for $a_1,a_2 \in \R$ and observe that 
\begin{eqnarray}\label{lc01}
E\left(\frac{1}{n}||\y||^2\right) & = & \s^2 + \tau^2 \\ \label{lc02}
E\left\{\frac{1}{n}\y^TX(X^TX)^{-1}X^T\y\right\} & = & \frac{d}{n}\s^2 + \tau^2
\end{eqnarray}
are non-redundant linear combinations of $\s^2$ and $\tau^2$. Since 
\begin{eqnarray*}
EL_0(a_1,a_2) & = & a_1E\left(\frac{1}{n}||\y||^2\right)  + a_2
E\left\{\frac{1}{n}\y^TX(X^TX)^{-1}X^T\y\right\} \\
& = & a_1(\s^2 + \tau^2) + a_2\left(\frac{d}{n}\s^2 + \tau^2\right),
\end{eqnarray*}
it follows that there exist $a_{11}, a_{12} \in \R$ such that
$L_0(a_{11},a_{12})$ is an unbiased estimator of $\s^2$,
i.e. $EL_0(a_{11},a_{12}) = \s^2$.  In particular, we have
\[
EL_0\left(\frac{n}{n-d},-\frac{n}{n-d}\right) = \s^2
\]
and, moreover, $\hat{\s}_0^2 = L_0\{n/(n-d),-n/(n-d)\}$.  Thus, for $d
< n$, 
$\hat{\s}^2_0$ may be viewed as the unique linear combination of
$n^{-1}||\y||^2$ and $n^{-1}\y^TX(X^TX)^{-1}X^T\y$ that yields an unbiased
estimator of $\s^2$.   

The identities
(\ref{lc01})-(\ref{lc02}) also imply that there exist $a_{21},a_{22} \in \R$
such that $L_0(a_{21},a_{22})$ is an unbiased estimator for $\tau^2$.
Indeed, 
\[
EL_0\left(-\frac{d}{n-d},\frac{n}{n-d}\right) = \tau^2
\]
and
\[
\hat{\tau}_0^2 =L_0\left(-\frac{d}{n-d},\frac{n}{n-d}\right)
\]
is the estimator defined initially in (\ref{tau0}). 

The ideas above are easily adapted to a more
general setting that is useful for problems where $d \geq n$.  Broadly, we seek statistics $T_1 =
T_1(\y,X)$ and $T_2 = T_2(\y,X)$ such that 
\begin{equation} \label{ET}\begin{array}{rcl}
E(T_1) & = & b_{11}\s^2 + b_{12}\tau^2 \\
E(T_2) & =& b_{21}\s^2 + b_{22}\tau^2 \end{array} \ \ \mbox{for some
constants}\ \ \begin{array}{l} b_{11},b_{12}\\ b_{21},b_{22}\end{array}\!\! \in \R \ \ \mbox{with} \  
\ b_{11}b_{22}-b_{12}b_{21} \neq 0.
\end{equation}
In other words, the expected value of the statistics $T_1$, $T_2$
should form a pair of non-degenerate linear combinations of $\s^2$ and
$\tau^2$.  If such $T_1$ and $T_2$ can be found,
then unbiased estimators for
$\s^2$, $\tau^2$ may be formed by taking linear combinations of $T_1$
and $T_2$.  Moreover, asymptotic properties of these estimators are determined by the asymptotic properties of $T_1$,
$T_2$.

In the example discussed above, where $d < n$, $T_1 = n^{-1}||\y||^2$ and
$T_2 = n^{-1}\y^TX(X^TX)^{-1}X^T\y$.    If $d \geq n$,
then alternatives to $T_2 = n^{-1}\y^TX(X^TX)^{-1}X^T\y$ must be
sought; in this paper, we focus on $T_2 = n^{-2}||X^T\y||^2$ (remarks
on other potential alternatives may be found in Section 5).  Using basic facts about the Wishart distribution (see
Supplemental Text for formulas involving various moments of the
Wishart distribution, which are obtained using techniques
from \citep{letac2004all, graczyk2005hyperoctahedral} and are used
throughout the paper), we have
\begin{eqnarray} \nonumber
E\left(\frac{1}{n^2}||X^T\y||^2\right) & = & \frac{1}{n^2}E\y^TXX^T\y
\\ \nonumber
& =& \frac{1}{n^2}E\b^T(X^TX)^2\b + \frac{1}{n^2}E\ee^TXX^T\ee \\ \label{ET2}
& = & \frac{d + n + 1}{n}\tau^2 + \frac{d}{n}\s^2.
\end{eqnarray}
Since $E(n^{-1}||\y||^2) = \s^2 + \tau^2$, it follows that $T_1 =
n^{-1}||\y||^2$ and $T_2 = n^{-2}||X^T\y||^2$ satisfy (\ref{ET}).
Moreover, $T_2 = n^{-2}||X^T\y||^2$ is defined and (\ref{ET2}) is
valid even when $d \geq n$.  Now let
\[
L(a_1,a_2) = \frac{a_1}{n}||\y||^2 + \frac{a_2}{n^2}||X^T\y||^2.
\]
and define 
\begin{eqnarray*}
\hat{\s}^2 & = & L\left(\frac{d + n + 1}{n + 1}, -
  \frac{n}{n+1}\right) \ \ = \ \ \frac{d + n + 1}{n(n+1)}||\y||^2 -
\frac{1}{n(n+1)}||X^T\y||^2 \\
\hat{\tau}^2 & = & L\left(-\frac{d}{n+1},\frac{n}{n+1}\right) \ \ = \
\ -\frac{d}{n(n+1)}||\y||^2 + \frac{1}{n(n+1)}||X^T\y||^2.
\end{eqnarray*}
Making use of (\ref{lc01}) and (\ref{ET2}), a basic calculation implies that
$\hat{\s}^2$ and $\hat{\tau}^2$ are unbiased estimators for $\s^2$ and
$\tau^2$.  Thus, we have the following theorem.  

\begin{thm} {\em [Unbiasedness]} Suppose that $\S = I$.  Then
$E(\hat{\s}^2) = \s^2$ and $E(\hat{\tau}^2) = \tau^2$. 
\end{thm}

\subsection{Consistency}

Let $\hat{\bth} = (\hat{\s}^2,\hat{\tau}^2)^T$ and let $\T = (n^{-1}||\y||^2,n^{-2}||X^T\y||^2)^T$. The covariance matrix of $\hat{\bth}$ is important
for understanding the asymptotic properties of $\hat{\s}^2$ and
$\hat{\tau}^2$.  Since $\hat{\bth} = A\T$, where
\begin{equation}\label{matA}
A = \left(\begin{array}{cc} \frac{d + n + 1}{n+1} & -\frac{n}{n+1} \\
    -\frac{d}{n+1} & \frac{n}{n+1} \end{array}\right),
\end{equation}
it follows that $\Cov(\hat{\bth}) = A\Cov(\T)A^T$.  The covariance
matrices for $\hat{\bth}$ and $\T$ are both computed explicitly in the
Appendix.  Asymptotic approximations for the entries of $\Cov(\hat{\bth})$ that are valid as $d/n \to \rho \in [0,\infty)$ are given below:
\begin{eqnarray}\label{v1}
\Var(\hat{\s}^2) & \sim & 
\frac{2}{n}\left\{\rho(\s^2 + \tau^2)^2 + \s^4 + \tau^4\right\}\\ \label{v2}
\Var(\hat{\tau}^2) & \sim & \frac{2}{n}\left\{(\rho
  + 1)(\s^2 + \tau^2)^2 - \s^4 + 3\tau^4\right\}\\ \label{v3}
\Cov(\hat{\s}^2,\hat{\tau}^2) & \sim & 
-\frac{2}{n}\left\{\rho(\s^2 + \tau^2)^2 + 2\tau^4\right\}.
\end{eqnarray}
The following theorem contains a slightly more detailed version of
these approximations, and gives an explicit consistency result for $\hat{\s}^2$,
$\hat{\tau}^2$.  The theorem is proved in the Appendix.  

\begin{thm}{\em [Consistency]} Suppose that $\S = I$.  Then
\begin{eqnarray*}
\Var(\hat{\s}^2) & = & \frac{2}{n}\left\{\frac{d}{n}(\s^2 + \tau^2)^2
  + \s^4 + \tau^4\right\}\left\{1 +
  O\left(\frac{1}{n}\right)\right\} \\
\Var(\hat{\tau}^2) & = & \frac{2}{n}\left\{\left(1 +
    \frac{d}{n}\right)(\s^2 + \tau^2)^2 - \s^4 + 3\tau^4\right\}\left\{1 +
  O\left(\frac{1}{n}\right)\right\} \\
\Cov(\hat{\s}^2,\hat{\tau}^2) & = & -\frac{2}{n}\left\{\frac{d}{n}(\s^2
+ \tau^2)^2 + 2\tau^4\right\}\left\{1 +
  O\left(\frac{1}{n}\right)\right\}.
\end{eqnarray*}
In particular, 
\[
|\hat{\s}^2 - \s^2|, \ |\hat{\tau}^2 - \tau^2| = O_P\left\{\sqrt{\frac{d +
    n}{n^2}}(\s^2 + \tau^2)\right\}.
\]
\end{thm}

\begin{rk}
If $d/n \to \rho \in [0,\infty)$, then the asymptotic approximations
(\ref{v1})-(\ref{v3}) follow immediately from Theorem 2.
\end{rk}

\begin{rk}
It is instructive to compare the asymptotic variance and covariance of
$\hat{\s}^2$, $\hat{\tau}^2$ to that of the estimators $\hat{\s}_0^2$,
$\hat{\tau}_0^2$, defined in (\ref{s0})-(\ref{tau0}).  If $n \to \infty$ and $d/n \to \rho \in [0,1)$,
then
\begin{eqnarray*}
\Var(\hat{\s}^2_0) & \sim & \frac{2\s^4}{n(1 - \rho)} \\
\Var(\hat{\tau}^2_0) & \sim & \frac{2}{n}\left\{(\s^2 + \tau^2)^2 + \left(\frac{\rho}{1 - \rho} - 1\right)\s^4\right\} \\
\Cov(\hat{\s}_0^2,\hat{\tau}_0^2) & \sim & -\frac{2\rho\s^4}{n(1 - \rho)}.
\end{eqnarray*}
Notice that in (\ref{v1}), $\Var(\hat{\s}^2)$ increases with the signal strength
$\tau^2$, while $\Var(\hat{\s}_0^2)$ does not depend on $\tau^2$.  On
the other hand, $\Var(\hat{\s}^2) < \Var(\hat{\s}_0^2)$ when $\tau^2$
is small or $\rho$ is close to 1.  
\end{rk}

\begin{rk}
Suppose that $c_1,c_2 > 0$ are fixed.  Theorem 2 implies that if $d/n \to \rho
\in [0,\infty)$, then $\hat{\s}^2$, $\hat{\tau}^2$ are consistent in
the sense that
\begin{equation}\label{rka}
\lim_{d/n \to \rho}\  \sup_{\substack{0 \leq \s^2 < c_1 \\ 0 \leq
    \tau^2 < c_2
  }} E(\hat{\s}^2 - \s^2)^2 = \lim_{d/n \to \rho} \ \sup_{\substack{0
    \leq \s^2 < c_1 \\ 0  \leq \tau^2 < c_2 }}
E(\hat{\tau}^2 - \tau^2)^2 = 0.
\end{equation}
On the other hand, \cite{dicker2012optimal} proved that if $\rho > 0$,
then it is
impossible to estimate $\bb$ in this setting.  In particular, if $\rho
> 0$, then 
\[
\liminf_{d/n \to \rho} \ \inf_{\hat{\bb}} \ \sup_{\substack{0 \leq \s^2 < c_1 \\
    0 \leq \tau^2 < c_2}} E||\hat{\bb} - \bb||^2 > 0,
\]
where the infimum is over all measurable estimators for $\bb$. Thus,
Theorem 2 describes methods for consistently estimating
$\s^2$ and $\tau^2$ in high-dimensional
linear models, where it is impossible to estimate $\bb$.  If $\rho \in
[0,1)$, then (\ref{rka}) holds with $\hat{\s}_0^2$, $\hat{\tau}_0^2$
in place of $\hat{\s}^2$, $\hat{\tau}^2$. However, Theorem 2 also
applies to settings where $d > n$ (i.e. $\rho > 1$) and the estimators
$\hat{\s}_0^2$, $\hat{\tau}_0^2$ are undefined.  \hfill $\Box$
\end{rk}

\subsection{Asymptotic normality}

Define the total variation distance
between random variables $u$ and $v$ to be
\[
d_{TV}(u,v) = \sup_{B \in \mathcal{B}(\R)} |P(u \in B) - P(v \in B)|,
\]
where $\mathcal{B}(\R)$ denotes the collection of Borel sets in $\R$.
The next theorem is this paper's main result on asymptotic normality.
It is a direct application of results in
\citep{chatterjee2009fluctuations}.  Theorem 3 is proved in the Appendix and it
is valid for arbitrary positive definite covariance matrices $\S$.

\begin{thm} {\em [Asymptotic normality]} 
Let $\l_1 = ||n^{-1}X^TX||$
be the operator norm of $n^{-1}X^TX$ (i.e. $\l_1$ is the largest
eigenvalue of $n^{-1}X^TX$).   Let $h: \R^2 \to \R$ be a function with continuous second order partial
derivatives, let $\nabla h$ denote the gradient of $h$, and let
$\nabla^2h$ denote the Hessian of $h$.  Suppose that $\psi^2 =
\Var\{h(\T)\} < \infty$ and let $w$ be a normal random
variable with the same mean and variance as $h(\T)$.  Then
\begin{equation}\label{thm3bd}
d_{TV}\{h(\T),w\} =  O\left(\frac{||\S||^{3/2}\xi\nu}{n^{3/2}\psi^2}\right),
\end{equation}
where $\xi$ and $\eta$ are
defined as follows:
\begin{eqnarray*}
\xi & = & \xi(\s^2,\tau^2,\S,d,n) \ \ = \ \ \g_4^{1/4} + \g_2^{1/4} +
\g_0^{1/4}\tau(\tau + 1) \\
\nu & = & \nu(\s^2,\tau^2,\S,d,n) \ \ = \ \ \eta_8^{1/4} +
\eta_4^{1/4} + \eta_0^{1/4}\tau^2(\tau^2 + 1) + \g_4^{1/4} +
\g_0^{1/4}(\tau^2 + 1)
\end{eqnarray*}
and, for non-negative integers $k$, 
\begin{eqnarray*}
\g_k & = & \g_k(\s^2,\tau^2,\S,d,n) \ \ = \ \ E\left\{\left|\left|\nabla h(\T)\right|\right|^4(\l_1 +
  1)^6\left(\frac{1}{n}||\ee||^2\right)^k\right\}, \\
\eta_k & = & \eta_k(\s^2,\tau^2,\S,d,n) \ \ = \ \  E\left\{\left|\left|\nabla^2 h(\T)\right|\right|^4(\l_1 +
  1)^{12}\left(\frac{1}{n}||\ee||^2\right)^k\right\}.
\end{eqnarray*}
\end{thm}

\setcounter{rk}{0}

\begin{rk}
If $||\S||$ is bounded, then the asymptotic behavior of the upper
bound (\ref{thm3bd}) is determined by that of $\xi$, $\nu$, and
$\psi^2$, which, in turn, is determined by the function $h$.  For the functions $h$ considered in this paper, if $d/n \to
\rho \in [0,\infty)$, then $\xi$, $\nu$,
and $n\psi^2$ are bounded by rational functions in $\s^2$ and
$\tau^2$.  Thus, if $||\S||$ is bounded, $d/n \to \rho \in [0,\infty)$, 
and $\s^2,\tau^2$ lie in some compact set, then we typically have
\[
d_{TV}\{h(\T),w\} = O(n^{-1/2}).
\]  
In other words, $h(\T)$ converges to a normal random variable at rate
$n^{-1/2}$.  Under these conditions, if $\psi^2 = \Var\left\{h(\T)\right\}$ is known or
estimable (as it is for the $h$ studied here), then asymptotically
valid confidence intervals for $Eh(\T)$ may be constructed using
Theorem 3.  \hfill $\Box$
\end{rk}

Now let $A$ be the matrix (\ref{matA}) and let $\mathbf{a}_1^T$, $\mathbf{a}_2^T$ denote the
first and second rows of $A$, respectively.  Applying Theorem 3 with
$\S = I$ and $h(\T) = \mathbf{a}_1^T\T = \hat{\s}^2$, $h(\T) = \mathbf{a}_2^T\T =
\hat{\tau}^2$, and $h(\T) = (\mathbf{a}_2^T\T)/(\mathbf{a}_1^T\T) =
\hat{\tau}^2/\hat{\s}^2$  gives bounds on the total variation distance between
$\hat{\s}^2$, $\hat{\tau}^2$, and $\hat{\tau}^2/\hat{\s}^2$ and corresponding
normal random variables.  These examples are pursued in more detail below.

\begin{ex}[$\hat{\s}^2$ and $\hat{\tau}^2$]
Let $h(\T) = \mathbf{a}_1^T\T = \hat{\s}^2$ in Theorem 3 and suppose
that $\S = I$.  Then $\eta_k =
0$, because $\nabla^2 h = 0$.  To bound $\g_k$, we have
\[
\g_k  =  E\left\{||\mathbf{a}_1||^4(\l_1 +
  1)^6\left(\frac{1}{n}||\ee||^2\right)^k\right\}  =  O\left\{\left(1 + \frac{d}{n}\right)^{10}\s^{2k}\right\}.
\]
Thus,
\begin{eqnarray*}
\xi & = & O\left\{\left(1 + \frac{d}{n}\right)^{5/2}\left(\s^2 + \s +
    \tau^2+ \tau\right)\right\} \\
\nu & = &  O\left\{\left(1 + \frac{d}{n}\right)^{5/2}\left(\s^2 +
    \tau^2 + 1\right)\right\}.
\end{eqnarray*}
By Theorem 2,
\[
\Var(\hat{\s}^2) = \frac{2}{n}\left\{\frac{d}{n}(\s^2 + \tau^2)^2 + \s^4 +
  \tau^4\right\}\left\{1 + O\left(\frac{1}{n} \right)\right\}.
\]
Now let 
\begin{equation}\label{psi1}
\psi_1^2 = 2 \left\{\frac{d}{n}(\s^2 + \tau^2)^2 + \s^4 +
  \tau^4\right\}
\end{equation}
and let $z \sim N(0,1)$. Then Theorem 3 implies 
\[
d_{TV}\left\{\sqrt{n}\left(\frac{\hat{\s}^2 -
      \s^2}{\psi_1}\right),z\right\}  =  O\left[\frac{1}{\sqrt{n}}\left( 1+
    \frac{d}{n}\right)^4\left\{1 + \left(\frac{1}{\s +
        \tau}\right)^3\right\}\right].  
\]
Similar calculations imply that 
\[
d_{TV}\left\{\sqrt{n}\left(\frac{\hat{\tau}^2 -
      \tau^2}{\psi_2}\right),z\right\} = O\left[\frac{1}{\sqrt{n}}\left( 1+
    \frac{d}{n}\right)^4\left\{1 + \left(\frac{1}{\s +
        \tau}\right)^3\right\}\right],
\]
where 
\begin{equation}\label{psi2}
\psi_2^2 = 2\left\{\left(1 + \frac{d}{n}\right)(\s^2 + \tau^2)^2 - \s^4
  + 3\tau^4\right\}. 
\end{equation}
Thus, we have the following corollary to
Theorem 3.
\end{ex}

\begin{cor}
Suppose that $\S = I$ and $D \subseteq (0,\infty)$ is compact.  Let $z
\sim N(0,1)$. If $d/n \to \rho \in [0,\infty)$, then 
\[
\sup_{\s^2,\tau^2 \in D} d_{TV}\left\{\sqrt{n}\left(\frac{\hat{\s}^2 -
      \s^2}{\psi_1}\right),z\right\}, \ \sup_{\s^2,\tau^2
    \in D} d_{TV}\left\{\sqrt{n}\left(\frac{\hat{\tau}^2 -
      \tau^2}{\psi_2}\right),z\right\} = O(n^{-1/2}),
\]
where $\psi_1,\psi_2$ are defined in (\ref{psi1})-(\ref{psi2}).  
\end{cor}

\begin{ex}[Signal-to-noise ratio]  Suppose that $\S = I$.  Define the
  function $g_0: \R^2
\setminus\{0\}\times \R \to  \R$ by $g_0(\u) = g_0(u_1,u_2) = u_2/u_1$
and let $h_0 = g_0 \circ A$ be defined by $h_0(\t) = g_0(A\t)$, where
$A$ is the $2 \times 2$ matrix given in (\ref{matA}).  Then $h_0(\T) =  g_0(\hat{\s}^2,\hat{\tau}^2) =
\hat{\tau}^2/\hat{\s}^2$ is an estimate of the
signal-to-noise ratio.  However, Theorem 3 cannot be applied
directly because $h_0$ is not defined on all of $\R^2$ (if
$\mathbf{a}_1^T\t = 0$, then $h_0(\t)$ is undefined).  To remedy
this, we assume that $\s^2,\tau^2 \in D$, where $D
\subseteq (0,\infty)$ is compact and, moreover, that $d/n \to \rho \in
[0,\infty)$.  Now let $g:\R^2 \to \R$ be a function with continuous
second order partial derivatives such that $\sup_{\u \in \R^2}
||\nabla g(\u)||, \ \sup_{\u \in \R^2} ||\nabla^2 g(\u)|| <
\infty$ and $g = g_0$ on $D_0 \times D_0$, where $D_0 \subseteq
(0,\infty)$ is a compact set containing $D$ in its interior.  

To show
that the estimated signal-to-noise ratio is asymptotically normal, we
apply Theorem 3 with $h = g \circ A$.   Working under the
assumption that $\s^2, \tau^2 \in D$ and
$d/n \to \rho \in [0,\infty)$, it is straightforward to check that $\g_k,\eta_k = O(1)$, for $k= 0,2,4,8$;
thus, $\xi, \nu = O(1)$.  To approximate the variance of $h(\T)$, let
$\bth = (\s^2,\tau^2)^T$ and $\hat{\bth} =
(\hat{\s}^2,\hat{\tau}^2)^T$.  A
second order Taylor expansion yields
\begin{eqnarray}\nonumber 
h(\T) & = & g(\hat{\bth}) \\ \label{ex2a}
& = & g(\bth) + \nabla g(\bth)^T(\hat{\bth} - \bth) + R||\hat{\bth} - \bth||^2,
\end{eqnarray}
where $R = O(1)$. Theorem 2 and a
straightforward calculation imply that
\begin{eqnarray*}
\Var\left\{\nabla g(\bth)^T\hat{\bth}\right\} 
& = &
\nabla g(\bth)^T\Cov(\hat{\bth})\nabla g(\bth) \\
& = &  \frac{2}{n\s^8}\left\{\left(1 + \frac{d}{n}\right)\left(\s^2 +
    \tau^2\right)^4 - \s^4(\s^2 + \tau^2)^2\right\}\left\{1 + O\left(\frac{1}{n}\right)\right\}.
\end{eqnarray*}
Since $\Var\left(||\hat{\bth} - \bth||^2\right) = O(n^{-2})$ and $R =
O(1)$, (\ref{ex2a}) implies
\[
\psi^2 = \Var\left\{h(\T)\right\} =  \frac{2}{n\s^8}\left\{\left(1 + \frac{d}{n}\right)\left(\s^2 +
    \tau^2\right)^4 - \s^4(\s^2 + \tau^2)^2\right\}\left\{1 + O\left(\frac{1}{n}\right)\right\}.
\]
Thus, Theorem 3 implies that
\begin{equation}\label{ex2b}
d_{TV}\left[\sqrt{n}\left\{\frac{h(\T) -
      Eh(\T)}{\psi_0}\right\},z\right] = O(n^{-1/2}),
\end{equation}
where $z \sim N(0,1)$ and
\begin{equation}\label{psi0}
\psi_0^2 = \frac{2}{\s^8}\left\{\left(1 + \frac{d}{n}\right)\left(\s^2 +
    \tau^2\right)^4 - \s^4(\s^2 + \tau^2)^2\right\}.
\end{equation}
Finally, in order to relate (\ref{ex2b}) directly to
$h_0(\T) = \hat{\tau}^2/\hat{\s}^2$ and the signal-to-noise ratio $\tau^2/\s^2$, notice that Theorem
2 implies 
\[
P\left\{h(\T) \neq \frac{\hat{\tau}^2}{\hat{\s}^2}\right\} = O\left(\frac{1}{n}\right)
\]
and equation (\ref{ex2a}) implies 
\[
Eh(\T) = \frac{\tau^2}{\s^2} + O\left(\frac{1}{n}\right).
\]
Combining these facts with (\ref{ex2b}), we obtain the following
result.
\end{ex}

\begin{cor}
Suppose that $\S = I$ and $D \subseteq (0,\infty)$ is compact.  Let $z
\sim N(0,1)$.  If $d/n \to \rho \in [0,\infty)$, then 
\[
\sup_{\s^2,\tau^2 \in D} d_{TV}\left\{\sqrt{n}\left(\frac{\hat{\tau}^2/\hat{\s}^2 -
      \tau^2/\s^2}{\psi_0}\right),z\right\} = O(n^{-1/2}),
\]
where $\psi_0^2$ is defined in (\ref{psi0}).  
\end{cor}

\section{Unknown $\S$}

In this section, we propose estimators for $\s^2$, $\tau^2$ for use when $\S$
is an unknown $d \times d$ positive definite matrix.  In Section 3.1, we
consider the case where a norm-consistent estimator for $\S$ is
available.  In this setting, consistent (and, under certain conditions,
asymptotically normal) estimators for $\s^2$, $\tau^2$ are obtained by
essentially transforming the problem to the $\S = I$ case.  In Section
3.2, we consider the case where a norm-consistent estimator for $\S$
is not available.  Here we derive alternative estimators for $\s^2$,
$\tau^2$ and these estimator are shown to be consistent
and asymptotically normal under additional conditions on $\S$ and $\b$.  

\subsection{Estimable $\S$}

An estimator $\hat{\S}$ for $\S$ is norm consistent if
$||\hat{\S}- \S|| \to 0$, where $||\hat{\S} - \S||$ is the operator
norm of $\hat{\S} - \S$ and the convergence holds in some appropriate
sense (e.g. convergence in probability or squared-mean).  In
high-dimensional data analysis where $d/n \to \rho > 0$,
the sample covariance matrix $n^{-1}X^TX$ is not a norm-consistent
estimator for $\S$; furthermore, in the absence of additional information about $\S$, it is
generally not possible to find a norm-consistent estimator for $\S$.
However, \cite{bickel2008regularized}, \cite{elkaroui2008operator},
\cite{cai2010optimal}, and others have shown that for wide classes of
matrices $\S$, norm-consistent estimators are available when $d/n \to
\rho > 0$.  Moreover, one can
reasonably envision situations in practice where pertinent prior
information about the population predictor covariance matrix $\S$ is available (so that a reliable estimator of $\S$
may be found), but there is little prior information about
$\beta$ (so that $\beta$ is not estimable and estimates of $\s^2$,
$\tau^2$ based on residual sums of squares $||\y - X\hat{\bb}||^2$ are
suspect).  \cite{li2010bayesian} discuss relevant examples
from genomics and fMRI with highly structured high-dimensional
predictors, though they focus on variable selection problems.    

Suppose that $\hat{\S}$ is a positive definite estimator for $\S$ and define
the estimators 
\begin{eqnarray*}
\hat{\s}^2(\hat{\S}) & = & \frac{d + n + 1}{n(n+1)} ||\y||^2 -
\frac{1}{n(n+1)}||(X\hat{\S}^{-1/2})^T\y||^2 \\
\hat{\tau}^2(\hat{\S}) & = & -\frac{d}{n(n+1)} ||\y||^2 + \frac{1}{n(n+1)}||(X\hat{\S}^{-1/2})^T\y||^2.
\end{eqnarray*}
Notice that $\hat{\s}^2 = \hat{\s}^2(I)$ and $\hat{\tau}^2 =
\hat{\tau}^2(I)$.  Now let $Z = (\z_1,...,\z_n)^T = X\S^{-1/2}$.  Then $\z_1,...,\z_n \iidsim
N(0,I)$ and all of the results from Section 2 apply to
the estimators $\hat{\s}^2(\S)$, $\hat{\tau}^2(\S)$, with $Z$, $\S^{1/2}\bb$ in
place of $X$, $\bb$, respectively.  Since
\begin{eqnarray}\nonumber
\hat{\s}^2(\hat{\S}) & = & \hat{\s}^2(\S)
-\frac{1}{n(n+1)}\left\{||(X\hat{\S}^{-1/2})^T\y||^2 -
  ||Z^T\y||^2\right\} \\ 
&  = & \hat{\s}^2(\S) +
O\left\{\frac{1}{n^2}||Z^T\y||^2||\S^{1/2}\hat{\S}^{-1}\S^{1/2} - I||\right\} \label{est1}
\end{eqnarray}
and
\begin{eqnarray} \nonumber
\hat{\tau}^2(\hat{\S}) & = & \hat{\tau}^2(\S)
+\frac{1}{n(n+1)}\left\{||(X\hat{\S}^{-1/2})^T\y||^2 -
  ||Z^T\y||^2\right\} \\ 
&  = & \hat{\tau}^2(\S) +
O\left\{\frac{1}{n^2}||Z^T\y||^2||\S^{1/2}\hat{\S}^{-1}\S^{1/2} -
  I||\right\}, \label{est2}
\end{eqnarray}
we conclude that if $||\S^{1/2}\hat{\S}^{-1}\S^{1/2} - I||$ is small, then asymptotic
properties of $\hat{\s}^2(\hat{\S})$ and $\hat{\tau}^2(\hat{\S})$
are determined by those of $\hat{\s}^2(\S)$ and $\hat{\tau}^2(\S)$.  This is
illustrated in the following proposition, which is a direct
consequence of (\ref{est1})-(\ref{est2}) and the results of Section
2.  

\begin{prop}
Let $\hat{\S}$ be a positive definite estimator for $\S$.  Suppose further
that $||\S||$, $||\S^{-1}||$, $||\hat{\S}||$, $||\hat{\S}^{-1}|| = O_P(1)$.  
\begin{itemize}
\item[(i)] {\em [Consistency]}  
\[
|\hat{\s}^2(\hat{\S}) - \s^2|, \ |\hat{\tau}^2(\hat{\S}) - \tau^2| =
O_P\left\{\left(\sqrt{\frac{d + n}{n^2}} + ||\hat{\S} -
    \S||\right)(\s^2 + \tau^2)\right\}.
\]
\item[(ii)] {\em [Asymptotic normality]}  Let $\psi_1$, $\psi_2$, and
  $\psi_0$ be as defined in (\ref{psi1}), (\ref{psi2}), and
  (\ref{psi0}).  Suppose that $d/n \to \rho \in [0,\infty)$ and that  $\s^2,\tau^2 \in D$ for some compact set $D \subseteq
  (0,\infty)$.  If $||\hat{\S} - \S|| =
  o_P(n^{-1/2})$, then
\[
\sqrt{n}\left\{\frac{\hat{\s}^2(\hat{\S}) - \s^2}{\psi_1}\right\},\ \
\sqrt{n}\left\{\frac{\hat{\tau}^2(\hat{\S}) - \tau^2}{\psi_2}\right\}, \ \
\sqrt{n}\left\{\frac{\hat{\tau}^2(\hat{\S})/\hat{\s}^2(\hat{\S}) - \tau^2/\s^2}{\psi_0}\right\}
 \leadsto N(0,1),
\]
where $\leadsto$ indicates convergence in distribution.  
\end{itemize}
\end{prop}

\setcounter{rk}{0}
\begin{rk}
Part (i) of Proposition 1 implies if $\s^2,\tau^2$ are bounded, $d =
o(n^2)$, and $||\hat{\S} - \S|| = o_P(1)$, then $\hat{\s}^2(\hat{\S})$ and
$\hat{\tau}^2(\hat{\S})$ are weakly consistent for $\s^2$ and
$\tau^2$, respectively.  
\end{rk}

\begin{rk}
If $||\hat{\S} - \S|| = o_P(n^{-1/2})$ and the other conditions of
Proposition 1 are met, then $\hat{\s}^2(\hat{\S})$,
$\hat{\tau}^2(\hat{\S})$, and
$\hat{\tau}^2(\hat{\S})/\hat{\s}^2(\hat{\S})$ are asymptotically
normal with the same asymptotic variance as $\hat{\s}^2(\S)$,
$\hat{\tau}^2(\S)$, and $\hat{\tau}^2(\S)/\hat{\s}^2(\S)$, respectively.  The
condition $||\hat{\S} - \S|| = o_P(n^{-1/2})$ is quite strong.
However, \cite{bickel2008regularized} and \cite{cai2010optimal}
describe broad classes of covariance matrices $\S$ that can be
estimated at this rate.  For concreteness, we note that if the entries of $\x_i$ follow
one of many common time series models (e.g. $\textsc{AR}(k)$ for fixed
$k$), then there exist estimators $\hat{\S}$ such that $||\hat{\S} -
\S|| = o_P(n^{-1/2})$ when $d/n \to \rho \in (0,\infty)$.  \hfill $\Box$
\end{rk}

\subsection{Non-estimable $\S$}

Define $\tau_k^2 = \bb^T\S^k\bb$ and $m_k = d^{-1}\tr(\S^k)$, $k =
0,1,2,...$.  Then $\tau^2 = \tau_1^2$.  For general positive definite matrices $\S$, one easily checks that
\begin{eqnarray}\label{gE1}
\frac{1}{n}E||\y||^2 & = &   \s^2 + \tau_1^2 \\\label{gE2}
\frac{1}{n^2}E||X^T\y||^2 & = & \frac{d}{n}m_1 \s^2 + \frac{d}{n}m_1\tau_1^2
+ \left(1 + \frac{1}{n}\right)\tau_2^2
\end{eqnarray}
and
\begin{eqnarray*}
E\hat{\s}^2 & = & \frac{d(1 - m_1) + n + 1}{n+1}\s^2 + \frac{d(1 - m_1) + n + 1}{n+1}\tau_1^2 - \tau_2^2 \\
E\hat{\tau}^2 & = & \frac{d(m_1 - 1)}{n+1}\s^2 + \frac{d(m_1 - 1)}{n+1}\tau_1^2 + \tau^2_2.
\end{eqnarray*}  
Thus, if $\S \neq I$, then $\hat{\s}^2$, $\hat{\tau}^2$ are typically
{\em not} unbiased
estimators for $\s^2$, $\tau^2$, respectively.
More generally, it follows that if $\S \neq I$, then the expected value of the linear
combination $L(a_1,a_2) = a_1n^{-1}||\y||^2 + a_2n^{-2}||X^T\y||^2$
typically depends on $\s^2$, $\tau_1^2$, $\tau_2^2$, and $\tr(\S)$.
By contrast, as seen in Section 2, if $\S = I$, then $\tau^2 =
\tau_1^2 = \tau_2^2$ and  $EL(a_1,a_2)$
is determined by $\s^2$ and $\tau^2$ (in addition to $a_1$, $a_2$,
$d$, $n$); indeed, in the $\S = I$ case, this fact is precisely what
is leveraged to obtain unbiased  estimators for $\s^2$, $\tau^2$.
This suggests that an alternative
method  for estimating $\s^2$, $\tau^2$ may be
necessary when $\S$ is unknown and non-estimable.

In this section, we do not completely abandon our strategy of estimating $\s^2$,
$\tau^2$ by using linear combinations of $n^{-1}||\y||^2$ and
$n^{-2}||X^T\y||^2$.  Rather, we propose modified versions of
$\hat{\s}^2$ and $\hat{\tau}^2$ that are consistent and asymptotically
normal, provided $\bb$ and $\S$ satisfy certain conditions that have
appeared previously in the random matrix theory literature.  These
conditions are stated below.
\begin{itemize}
\item[(A)] As $d \to \infty$, the empirical distribution
of the eigenvalues of $\S$ converges weakly to a 
probability 
distribution with support contained in a compact subset of
$(0,\infty)$ and cumulative distribution function $H$.
\item[(B)]  Let 
\[
M_k = \int x^k \ dH(x) \ \mbox{ and } \ \Delta_k = \left|\frac{1}{\tau_0^2}\bb^T\S^k\bb - M_k\right|,
\]
where the distribution $H$ is given in condition (A).
Then, as $d \to \infty$,  
\begin{equation}\label{B}\Delta_k \to 0, \  \ k = 1,2,3. 
\end{equation} 
\end{itemize}
Condition (A) is fairly standard and is frequently assumed to hold in asymptotic analyses in random
matrix theory \citep{marchenko1967distribution, bai2004clt,
  bai2007asymptotics, elkaroui2008spectrum}.  The compact support
requirement in condition (A) can likely be relaxed; however, this is
not pursued further here.   Condition (B) is
more specialized and requires that the parameter $\bb$ interacts with
$\S$ as determined by (\ref{B}). In fact, while condition (B) is
sufficient for our consistency results in this section, we
require a stronger version of condition (B) (stated precisely in
Proposition 2 (ii))
to obtain asymptotic normality.   \cite{bai2007asymptotics} and
\cite{pan2008central} have proposed conditions that are closely
related to (B) and the strengthened version of (B) appearing
in Proposition 2 (ii) (in
fact, their conditions are stronger, if $H$ has
finite moments).  \cite{bai2007asymptotics} have noted
that under condition (A), if $\S$ is an
independent, orthogonally invariant random matrix (e.g. if $\S$ is a
Wishart matrix and $E(\S) = c I$, for some constant $c > 0$), then
condition (B) holds for any $\bb$.  Furthermore,
\citep{bai2007asymptotics} point out that for any $\S$ there must exist some $\bb$ such
that condition (B) holds; for instance, take $\bb = \bar{\u}$, where
$\bar{\u} = n^{-1/2}(\u_1 +
\cdots + \u_d)$ and $\u_1,...,\u_d$ are orthonormal eigenvectors of
$\S$.  More broadly, (B) may be interpreted as requiring that $\bb$ and $\S$ are
asymptotically free.  

Presently, we provide a heuristic to motivate estimators for $\s^2$ and
$\tau^2$ under conditions (A) and (B).  Following the method of
moments, the identities
\[
\frac{1}{d}E\tr\left(\frac{1}{n}X^TX\right) = m_1 \mbox{ and }
\frac{1}{d}E\tr\left\{\left(\frac{1}{n}X^TX\right)^2\right\} =
\frac{d}{n}m_1^2 + \left(1 + \frac{1}{n}\right)m_2
\]
suggest that
\[
\hat{m}_1 = \frac{1}{d}\tr\left(\frac{1}{n}X^TX\right) \mbox{ and }
\hat{m}_2 = \frac{n}{d(n+1)}
\tr\left\{\left(\frac{1}{n}X^TX\right)^2\right\} - \frac{1}{d(n+1)}\tr\left(\frac{1}{n}X^TX\right)^2
\]
are reasonable estimators for $m_1$ and $m_2$, respectively.   Now assume that $d,n$ are large
and $d/n \approx \rho \in [0,\infty)$.  Then, for $k = 1,2$, condition
(A) implies that $\hat{m}_k \approx m_k
\approx M_k$ and (B) implies $\tau_k^2 \approx \tau^2\hat{m}_k/\hat{m}_1$.  Combining these approximations with equations
(\ref{gE1})-(\ref{gE2}) yields
\begin{eqnarray} \label{approx1}
\frac{1}{n}E||\y||^2 & = & \s^2 + \tau^2\\ 
\frac{1}{n^2}E||X^T\y||^2 & \approx & \frac{d}{n}\hat{m}_1\s^2 +
\left\{\frac{d}{n}\hat{m}_1 + \left(1 + \frac{1}{n}\right)\frac{\hat{m}_2}{\hat{m}_1}\right\}\tau^2.\label{approx2}
\end{eqnarray}
Observe that the right-hand side of (\ref{approx1})-(\ref{approx2})
consists of linear
combinations of $\s^2$ and $\tau^2$, with
coefficients determined by the known quantities $d$, $n$,
$\hat{m}_1$, and $\hat{m}_2$. Thus, we are able to obtain {\em nearly}
unbiased estimators of $\s^2$ and $\tau^2$ by taking linear
combinations of $n^{-1}||\y||^2$ and $n^{-2}||X^T\y||^2$, with
coefficients determined by $d$, $n$,
$\hat{m}_1$, and $\hat{m}_2$.  In particular,  define the estimators 
\begin{eqnarray*}
\tilde{\s}^2 & = & L\left\{1 + \frac{d \hat{m}_1^2}{(n+1)\hat{m}_2},
- \frac{n\hat{m}_1}{(n+1)\hat{m}_2}\right\} \\ 
&= & \left\{1 +
\frac{d \hat{m}_1^2}{(n+1)\hat{m}_2}\right\}\frac{1}{n}||\y||^2 -
\frac{\hat{m}_1}{n(n+1)\hat{m}_2}||X^T\y||^2 \\
\tilde{\tau}^2 & = &
L\left\{-\frac{d\hat{m}_1^2}{(n+1)\hat{m}_2},\frac{n\hat{m}_1}{(n+1)\hat{m}_2}\right\}
\\ & = & -\frac{d\hat{m}_1^2}{n(n+1)\hat{m}_2}||\y||^2 + \frac{\hat{m}_1}{n(n
  + 1)\hat{m}_2}||X^T\y||^2.
\end{eqnarray*}
A basic calculation using (\ref{approx1})-(\ref{approx2}) suggests
that $E(\tilde{\s}^2) \approx \s^2$ and $E(\tilde{\tau}^2) \approx
\tau^2$. 

Proposition 2 summarizes some asymptotic properties of
$\tilde{\s}^2$ and $\tilde{\tau}^2$.  An outline of the proof, which
is fairly straightforward, may be  found in the Appendix. 

\begin{prop} 
Suppose that condition (A) holds, that $D \subseteq (0,\infty)$ is a
compact set, and that $\s^2, \tau^2 \in D$.  Suppose further that
there exist constants $c_1, c_2$ in $\R$ such that either $0 < c_1 < d/n
< c_2 < 1$ or $1 < c_1 < d/n < c_2 < \infty$, and suppose that $|n - d|>
9$.   Define $\tilde{\Delta}_k = \Delta_1 + |m_1 - M_1| + \cdots +
\Delta_k + |m_k - M_k|$, where $M_j$ and $\Delta_j$ are defined in
condition (B), and $m_j = d^{-1}\tr(\S^j)$.  
\begin{itemize}
\item[(i)] {\em [Consistency]}
\[
E(\tilde{\s}^2 - \s^2)^2,\ E(\tilde{\tau}^2 - \tau^2)^2 =  O\left(\frac{1 + \tilde{\Delta}_3}{n} +
  \tilde{\Delta}_2^2\right). 
\]
Thus, if condition (B) holds, then $|\tilde{\s}^2  - \s^2|, \
|\tilde{\tau}^2 - \tau^2| \to 0$ in mean-square.  
\item[(ii)] {\em [Asymptotic normality]}  Suppose that condition (B)
  holds, with the additional requirement that $\tilde{\Delta}_2 =
  o(n^{-1/2})$, and let
\begin{eqnarray*}
\tilde{\psi}_1^2 & = & 2\left\{\left(\frac{dm_1^2}{nm_2} +
  \frac{m_1m_3}{m_2^2} - 1\right)(\s^2 + \tau^2)^2 + \left(2 - \frac{m_1m_3}{m_2^2}\right)\s^4+
\frac{m_1m_3}{m_2^2}\tau^4 \right\} \\
\tilde{\psi}_2^2 & = & 2\left\{\left(\frac{dm_1^2}{nm_2} + \frac{m_1m_3}{m_2^2}\right)(\s^2
  + \tau^2)^2- \frac{m_1m_3}{m_2^2}\s^4 + \left(2 +
    \frac{m_1m_3}{m_2^2}\right)\tau^4\right\} \\
\tilde{\psi}_0^2 & = & \frac{2}{\s^8}\left\{\left(\frac{dm_1^2}{nm_2}
    + \frac{m_1m_3}{m_2^2}
  \right)(\s^2 + \tau^2)^4 -
\frac{m_1m_3}{m_2^2}\s^4(\s^2 + \tau^2)^2 \right. \\
&& \quad \quad  \left. - \left(1- \frac{m_1m_3}{m_2^2}
\right)\tau^4(\s^2 + \tau^2)^2\right\}.
\end{eqnarray*}
Then
\[
\sqrt{n}\left(\frac{\tilde{\s}^2 - \s^2}{\tilde{\psi}_1}\right), \
\sqrt{n}\left(\frac{\tilde{\tau}^2 - \tau^2}{\tilde{\psi_2}}\right), \ \sqrt{n}\left(\frac{\tilde{\tau}^2/\tilde{\s}^2 - \tau^2/\s^2}{\tilde{\psi_0}}\right)
  \leadsto N(0,1).  
\]
\end{itemize}
\end{prop}

\setcounter{rk}{0}

\begin{rk}
The conditions in Proposition 2 that require $|n - d|> 9$ and $d/n$ to
be bounded away from 1
are related to the fact that $\hat{m}_2^{-1}$ appears in both $\tilde{\s}^2$ and
$\tilde{\tau}^2$.  In particular, the mean-squared error of
$\tilde{\s}^2$ and $\tilde{\tau}^2$ may be
infinite if $n - d$ is not large enough.
\end{rk}

\begin{rk}
The condition $\tilde{\Delta}_2 = o(n^{-1/2})$ in part (ii) of
Proposition 2 is quite strong.  For instance, if $\S$ is a sample
covariance matrix formed from iid $N(0,\s_0^2)$ data with a constant
aspect ratio, then condition
(B) is satisfied, but $\tilde{\Delta}_2 \neq o(n^{-1/2})$.  On the
other hand, if $\S$ is a constant multiple of the identity matrix,
then $\tilde{\Delta}_2 = o(n^{-1/2})$.  We emphasize that only
conditions (A) and (B) are required for $\tilde{\s}^2$ and
$\tilde{\tau}^2$ to be consistent; $\tilde{\Delta}_2 = o(n^{-1/2})$ is
required for asymptotic normality.  
\end{rk}

\begin{rk}
If $\S = I$, then $m_1 = m_2 = m_3 = 1$ and $\tilde{\psi}_j^2 =
\psi_j^2$, $j = 0,1,2$, where $\psi_j^2$ are given in
(\ref{psi1})-(\ref{psi2}) and (\ref{psi0}).  In other words, if $\S = I$,
then the asymptotic variance of $\tilde{\s}^2$, $\tilde{\tau}^2$, and
$\tilde{\tau}^2/\tilde{\s}^2$ is the same as that of $\hat{\s}^2$, $\hat{\tau}^2$, and $\hat{\tau}^2/\hat{\s}^2$, respectively.
This is driven by the fact that if $d/n \to \rho \in (0,\infty)$, then
$|\hat{m}_k - m_k|$ converges at rate $n^{-1}$. \hfill $\Box$
\end{rk}

\section{Numerical results}

In this section, we study the performance of the proposed estimators
for $\s^2$, $\tau^2$, and the signal-to-noise ratio $\tau^2/\s^2$ via
simulation.  We consider three examples.  In the first example, we
report the results of a simulation study that illustrates the 
performance of the estimators from Section 2 (for $\S = I$) and Section
3.2 (unknown, non-estimable $\S$); the predictors
$\x_i$ are generated from various distributions (including
non-normal distributions) that are described below.  In the second
example, we compare the performance of $\hat{\s}^2 = \hat{\s}^2(I)$ to that of
$\hat{\s}_0^2 = (n-d)^{-1}||\y - X\hat{\bb}_{ols}||^2$ in settings
where $d < n$.  In the final example, we compare the performance of
estimators proposed in this paper to that of the scaled lasso and MC+
estimators for $\s^2$.  These estimators for $\s^2$ were proposed by
\cite{sun2011scaled} for settings where $\bb$ is sparse; in our
simulation study, we consider cases where $\bb$ is sparse and non-sparse.

\subsection{Example 1}

In this example, $d = 1000$ and the predictors $\x_i \in \R^{1000}$ were generated
according to one of three distributions.  In the first setting, $\x_i
\iidsim N(0,I)$.  In the second setting, we generated a $(2d) \times d \
\ (2000 \times
1000)$ random matrix $Z$ with iid $N(0,1)$ entries and took $\S =
(2d)^{-1}Z^TZ$; the iid predictors $\x_i$ were then generated according to
a $N(0,\S)$ distribution (the same matrix $\S$ was used for all datasets
generated under this setting).  In the third setting, the individual
predictors $x_{ij}$, $i = 1,...,n$, $j = 1,...,d$, were iid random
variables taking values in $\{\pm1\}$ with $P(x_{ij} = 1) = P(x_{ij} =
-1) = 0.5$.  

To generate the parameter $\bb \in \R^{1000}$, we created a
1000-dimensional vector with the first $d/2 = 500$ coordinates iid
$\mbox{uniform}(0,1)$ and the remaining $d/2 = 500$
coordinates iid $N(0,1)$;  $\bb$ was obtained by
standardizing this vector so that $||\bb||^2 = \tau_0^2 = 1$
(the same $\bb$ was used for all simulated datasets in this example).
The residual variance was fixed at $\s^2 = 1$ and we considered
datasets with $n = 500$ and $n = 1000$ observations.    

For each setting in this example, we generated 500 independent
datasets and computed the estimators $\hat{\s}^2 = \hat{\s}^2(I)$,
$\hat{\tau}^2 = \hat{\tau}^2(I)$,
$\hat{\tau}^2/\hat{\s}^2 = \hat{\tau}^2(I)/\hat{\s}^2(I)$ and
$\tilde{\s}^2$, $\tilde{\tau}^2$, $\tilde{\tau}^2/\tilde{\s}^2$ (the
estimators proposed in Section 2 and Section 3.2, respectively) for
each dataset.  Recall that the estimators from Section 2 were derived
under the assumption that $\x_i \sim N(0,I)$ and the estimators from
Section 3.2 were derived under the assumption that $\x_i \sim
N(0,\S)$, where $\S$ satisfies conditions (A)-(B).  Summary statistics
for the various estimators are reported in Table 1.  

\begin{center}
\begin{table}[h]
\begin{tabular}{ll||r|r||r|r||r|r}
& & \multicolumn{2}{|c||}{$\x_i \sim N(0,I)$}  & \multicolumn{2}{|c||}{$\x_i \sim N(0,\S)$} & \multicolumn{2}{|c}{$\x_i \in
\{\pm 1\}$ binary} \\ 
Estimator & $n$ & Mean & Std. Error & Mean & Std. Error & Mean &
Std. Error \\ \hline
$\hat{\s}^2(I)$ & 500 & 1.0118 & 0.1999 (0.2000) & 0.5552 & 0.2839 &1.0079 & 0.1976 \\
& 1000 &1.0003 & 0.1092 (0.1095) & 0.5428 & 0.1576 & 1.0035 & 0.1076 \\ \hline
$\tilde{\s}^2$ & 500 & 1.0120 & 0.2005 (0.2000) & 1.0283 & 0.1832 & 1.0039 &
0.1984 \\
& 1000 & 1.0003  & 0.1096 (0.1095) & 1.0237 & 0.1017 & 1.0014 & 0.1077  \\ \hline \hline 
$\hat{\tau}^2(I)$ & 500 & 0.9847 & 0.2364 (0.2366) & 1.4182 & 0.3396  & 0.9937 &
0.2442 \\
& 1000 & 0.9986 & 0.1408 (0.1414) & 1.4408 & 0.2007 & 1.0015 & 0.1402 \\ \hline
$\tilde{\tau}^2$ & 500 & 0.9846  & 0.2366 (0.2366) & 0.9450 & 0.2261 & 0.9977 & 0.2452  \\
& 1000 & 0.9986  & 0.1410 (0.1414) & 0.9600  & 0.1335 & 1.0036 & 0.1403 \\ \hline \hline 
$\hat{\tau}^2(I)/\hat{\s}^2(I)$ & 500 & 1.0687 & 0.5329 (0.4195) & 1.5685 &
22.6089 & 1.0801 & 0.5262 \\
& 1000 & 1.0234 & 0.2531 (0.2366) & 3.0488  & 2.5593 & 1.0212 & 0.2415  \\ \hline
$\tilde{\tau}^2/\tilde{\s}^2$ & 500 & 1.0694 & 0.5371 (0.4195) & 0.9881 & 0.4315 &
1.0901 & 0.5343 \\
& 1000 & 1.0236 & 0.2538 (0.2366) & 0.9573 & 0.2209 & 1.0256 & 0.2426 \\ 
\end{tabular}
\caption{Summary statistics for Example 1 ($d =  1000$).  Means and standard errors
  of various estimators, computed over 500 independent datasets for
  each configuration.  In each setting, $\s^2 = \tau^2 =
  \tau^2/\sigma^2 = 1$; thus, unbiased estimators should have mean
  close to 1. In the standard error column corresponding to $\x_i \sim
  N(0,I)$, numbers
  in parentheses are theoretically predicted standard errors (denoted
  $\psi_1$, $\psi_2$, and $\psi_0$ in the text;
  see Corollaries 1-2 and Proposition 2).  Theoretically predicted
  standard errors for $\x_i \sim N(0,\S)$ and $\x_i \in
  \{\pm 1\}$ binary are not known; more details may be found in the
  discussion in Section 4.1.}
\end{table}
\end{center}

One of the more striking aspects of the results reported
in Table 1 is the consistency and robustness of the estimators
$\tilde{\s}^2$, $\tilde{\tau}^2$, and $\tilde{\tau}^2/\tilde{\s}^2$.
Proposition 2 suggests that these estimators might be expected to
perform well when $\x_i \sim
N(0,I)$ and $\x_i \sim N(0,\S)$; none of our theoretical results apply
to the case where $\x_i \in \{\pm1\}$ is binary.  In the settings where
$\Cov(\x_i) = I$ and $\x_i \in \{\pm1\}$ is binary, the performance of $\tilde{\s}^2$, $\tilde{\tau}^2$, and $\tilde{\tau}^2/\tilde{\s}^2$ is nearly indistinguishable from that of
$\hat{\s}^2(I)$, $\hat{\tau}^2(I)$, and $\hat{\tau}^2(I)/\hat{\s}^2(I)$.  On the other hand, when
$\Cov(\x_i) = \S = (2d)^{-1}Z^TZ$, the estimators $\hat{\s}^2(I)$,
$\hat{\tau}^2(I)$, and $\hat{\tau}^2(I)/\hat{\s}^2(I)$ break down
significantly (their mean is far from the actual value $\s^2 = \tau^2 =
\tau^2/\s^2 = 1$), while $\tilde{\s}^2$, $\tilde{\tau}^2$, and
$\tilde{\tau}^2/\tilde{\s}^2$  still perform effectively.  The estimators
$\hat{\s}^2(I)$, $\hat{\tau}^2(I)$, and
$\hat{\tau}^2(I)/\hat{\s}^2(I)$ were developed under the assumption
that $\Cov(\x_i) = I$.  Thus, their diminished performance when
$\Cov(\x_i) \neq I$ is not unexpected.  The dramatically high standard
error 22.6089 for $\hat{\tau}^2(I)/\hat{\s}^2(I)$, when $\x_i \sim
N(0,\S)$ and $n = 500$ is indicative of instability when $\hat{\s}^2$ is
very small;  it also serves as a prompt to point out that our
estimators for $\s^2$ and $\tau^2$ can take both positive and negative
values.  Since $\s^2,\tau^2 \geq 0$, negative values for the
estimators may be undesirable.  In practice, one might choose to
implement special procedures for handling negative estimates of these
quantities; however, we take no such steps here.  In this example, the
only negative estimates of $\s^2$ and $\tau^2$ occurred for
$\hat{\s}^2(I)$ when $\x_i \sim N(0,\S)$: for $n = 500$, there were 18
datasets (out of 500) where $\hat{\s}^2(I) < 0$; for $n = 1000$, there
was one dataset where $\hat{\s}^2(I) < 0$.  

For $\x_i \sim N(0,I)$, Table 1 indicates that the empirical standard errors of the estimators
for $\s^2$ and $\tau^2$ are extremely close to the values predicted by
Corollary 1 and Proposition 2 (ii)  (denoted $\psi_1$ and $\psi_2$,
respectively; these values are displayed in parentheses in
Table 1). For the estimators of the signal-to-noise ratio $\tau^2/\s^2$, the
agreement between the empirical standard errors and the theoretically
predicted standard error $\psi_0$ (see Corollary 2 and Proposition
2 (ii)) is less compelling.   For $n = 500$, the empirical
standard errors for estimates of $\tau^2/\s^2$ are roughly 25\% larger
than the theoretically predicted standard errors.  For $n = 1000$, the empirical and
theoretical values are closer (they differ by approximately 10\%);
however, the discrepancy is still substantially larger than that for
estimates of $\s^2$ and $\tau^2$.  Figures 1 and 2 contain histograms
of the estimators for $\s^2$, $\tau^2$, and $\tau^2/\s^2$.  Normal
density plots with mean 1 (the actual value of $\s^2$, $\tau^2$, and
$\tau^2/\s^2$ in this example) and variance $\psi_1^2$, $\psi_2^2$,
and $\psi_0^2$ are superimposed on the histograms.  The histograms and
normal densities seem to agree quite well, as predicted by Corollaries
1-2 and Proposition 2.  

\begin{center}
\begin{figure}[h]
\includegraphics[width=\textwidth]{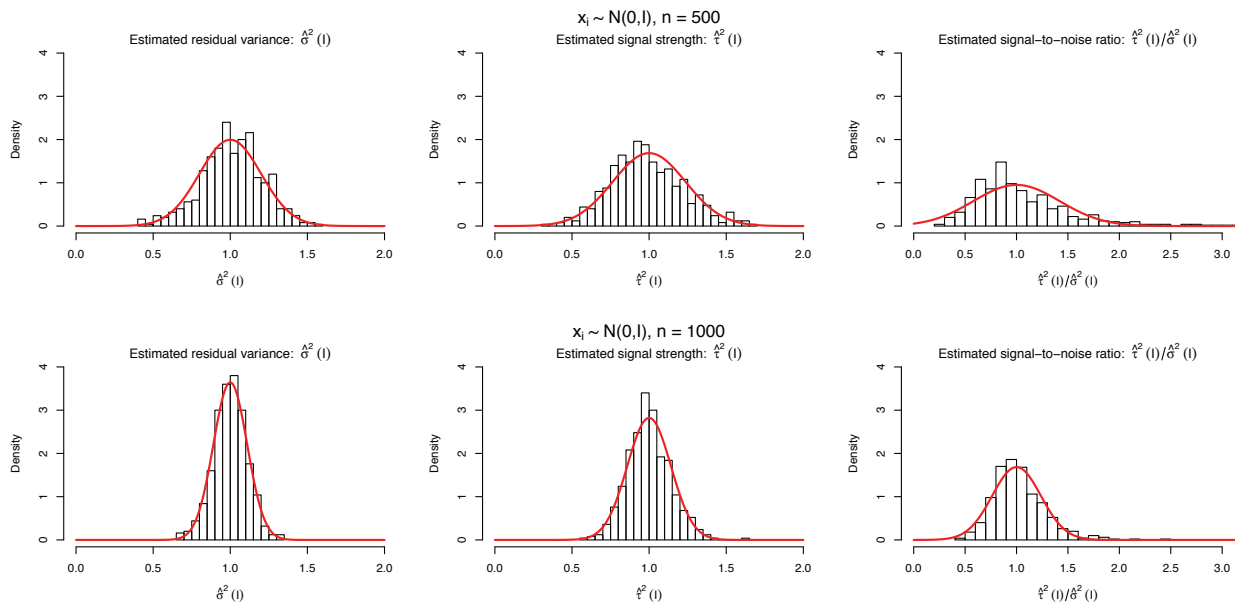}
\caption{Example  1 ($d = 1000$).  Histograms and normal density plots for the estimators
  $\hat{\s}^2(I)$, $\hat{\tau}^2(I)$, and
  $\hat{\tau}^2(I)/\hat{\s}^2(I)$, with $\x_i \sim N(0,I)$.  Top row, $n = 500$; bottom row, $n
  = 1000$.  Superimposed normal density plots have mean 1
  and variance $\psi_1^2$, $\psi_2^2$, and $\psi_0^2$ for
  $\hat{\s}^2(I)$, $\hat{\tau}^2(I)$, and
  $\hat{\tau}^2(I)/\hat{\s}^2(I)$, respectively. Corollaries 1-2 suggest that the distribution of the various estimators
  should be approximately equal to that of the corresponding normal distribution.}
\end{figure}
\end{center}
\begin{center}
\begin{figure}[h]
\includegraphics[width=\textwidth]{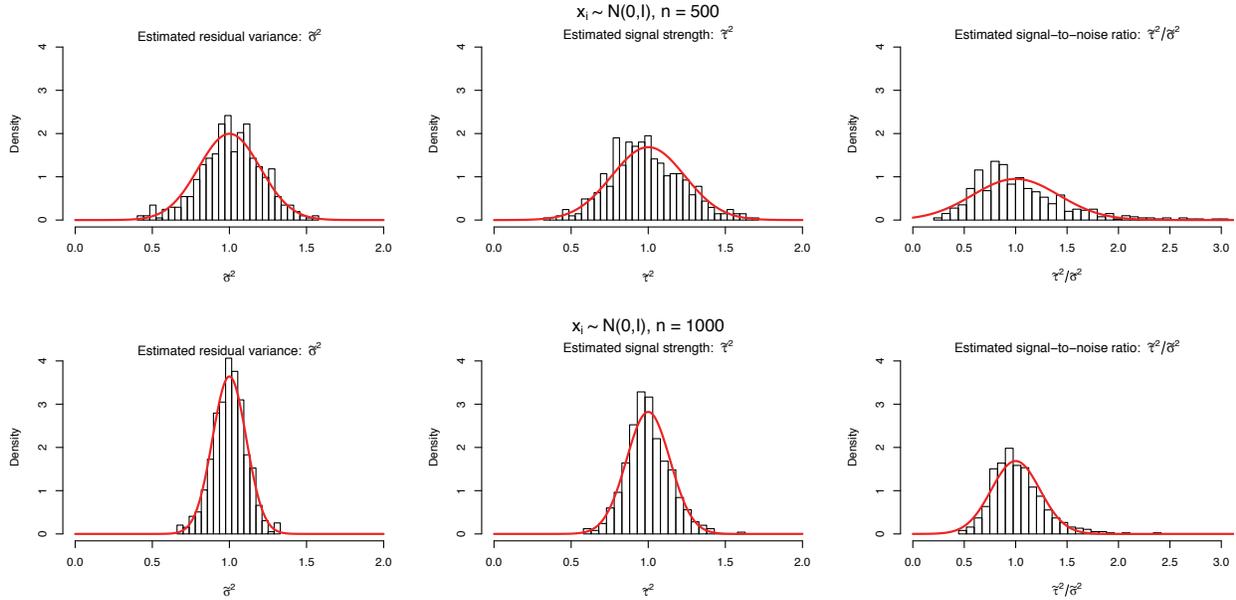}
\caption{Example  1 ($d = 1000$).  Histograms and normal density plots for the estimators
  $\tilde{\s}^2$, $\tilde{\tau}^2$, and
  $\tilde{\tau}^2/\tilde{\s}^2$, with $\x_i \sim N(0,I)$.  Top row, $n = 500$; bottom row, $n
  = 1000$.  Superimposed normal density plots have mean 1
  and variance $\psi_1^2$, $\psi_2^2$, and $\psi_0^2$ for
  $\tilde{\s}^2$, $\tilde{\tau}^2$, and
  $\tilde{\tau}^2/\tilde{\s}^2$, respectively. 
  Proposition 2 (ii) suggests that the distribution of the various estimators
  should be approximately equal to that of the corresponding normal distribution.}
\end{figure}
\end{center}

For $\x_i \sim N(0,\S)$, with $\S = (2d)^{-1}Z^TZ$, one might hope to use
Proposition 2 (ii)  to derive
theoretically predicted standard errors for the estimators
$\tilde{\s}^2$, $\tilde{\tau}^2$, and $\tilde{\tau}^2/\tilde{\s}^2$.  However, in order for Proposition 2 (ii) to apply, we must
have $\sqrt{n}\left|\bb^T\S^k\bb - ||\bb||^2d^{-1}\tr(\S^k)\right| \approx
0$, for $k = 1,2$.  In
this example, we had $\bb^T\S\bb = 0.9831$ and $\bb^T\S^2\bb =
1.4436$, while $||\bb||^2d^{-1}\tr(\S^2) = 1.0003$ and
$||\bb||^2d^{-1}\tr(\S) = 1.5018$.  Thus,
\begin{equation} \label{ex1a} \begin{array}{cc} 
\sqrt{500}\left|\bb^T\S\bb  - \dfrac{||\bb||^2}{d}\tr(\S)\right|  =
0.3839, & \sqrt{1000}\left|\bb^T\S\bb  -
  \dfrac{||\bb||^2}{d}\tr(\S)\right|  = 0.5429 \\ \\
\sqrt{500}\left|\bb^T\S^2\bb  - \dfrac{||\bb||^2}{d}\tr(\S^2)\right|  =
1.3002, & \sqrt{1000}\left|\bb^T\S^2\bb  -
  \dfrac{||\bb||^2}{d}\tr(\S^2)\right|  = 1.8387,  \end{array}
\end{equation}
which suggests that the applicability of Proposition 2 (ii) may be
questionable.  Moreover, asymptotically, if $\S = (2d)^{-1}Z^TZ$ and $d \to \infty$,
then the it is known that conditions of Proposition 2 (ii) are {\em not} satisfied (see Remark 2,
following Proposition 2).  Nevertheless, we believe it is
informative to compare the empirical distribution of the estimators
$\tilde{\s}^2$, $\tilde{\tau}^2$, and $\tilde{\tau}^2/\tilde{\s}^2$,
to normal distributions with mean 1 and variance $\tilde{\psi}_1^2$,
$\tilde{\psi}_2^2$, and $\tilde{\psi}_0^2$, respectively, as specified by
Proposition 2 (ii); corresponding histograms and normal density plots
may be found in Figure 3.   Upon visual inspection of Figure 3, the fit between
the sampling distribution of the estimators and the corresponding
normal distribution appears to be reasonably good.  The results in Table 1
indicate that there is slightly more bias in the estimators when $\x_i
\sim N(0,\S)$ than when $\x_i \sim N(0,I)$; this may be a result of the
discrepancies (\ref{ex1a}).
\begin{center}
\begin{figure}[h]
\includegraphics[width=\textwidth]{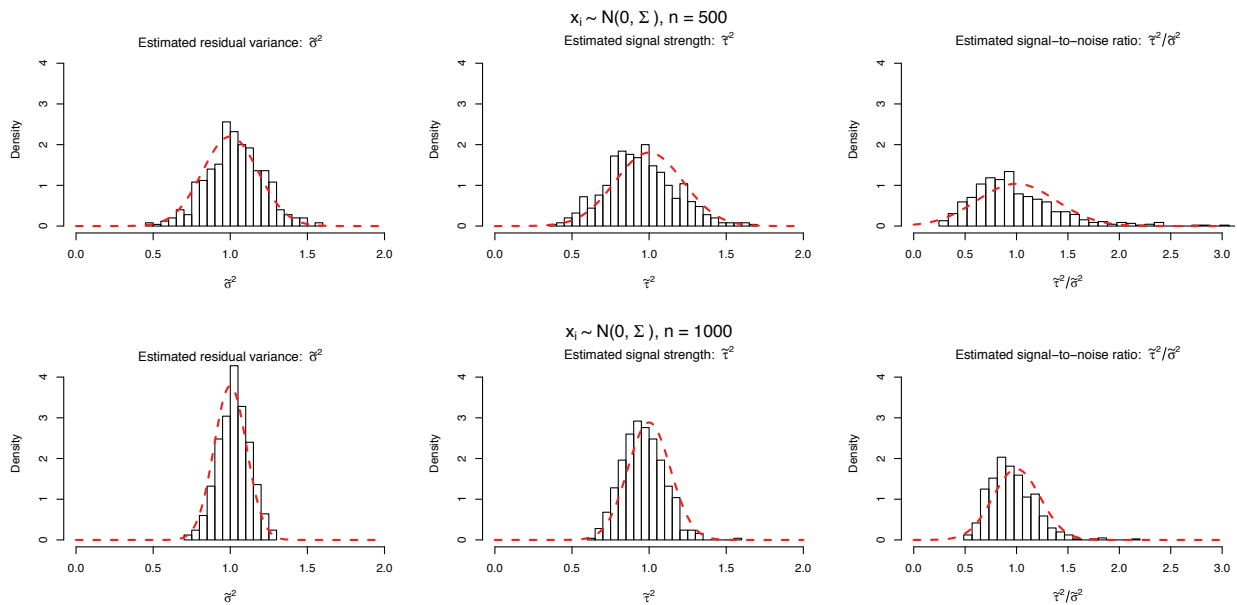}
\caption{Example  1 ($d = 1000$).  Histograms and normal density plots for the estimators
  $\tilde{\s}^2$, $\tilde{\tau}^2$, and
  $\tilde{\tau}^2/\tilde{\s}^2$, with $\x_i \sim N(0,\S)$ and $\S = (2d)^{-1}Z^TZ$.  Top row, $n = 500$; bottom row, $n
  = 1000$.  Superimposed normal density plots have mean 1
  and variance $\tilde{\psi}_1^2$, $\tilde{\psi}_2^2$, and $\tilde{\psi}_0^2$ for
  $\tilde{\s}^2$, $\tilde{\tau}^2$, and
  $\tilde{\tau}^2/\tilde{\s}^2$, respectively.  For $n = 500$,
  $\tilde{\psi}_1 = 0.1835$, $\tilde{\psi}_2 = 0.2211$, and
  $\tilde{\psi}_0 = 0.3841$; for $n = 1000$,
  $\tilde{\psi}_1 = 0.1054$, $\tilde{\psi}_2 = 0.1383$, and
  $\tilde{\psi}_0 = 0.2290$.  See Table 1 for empirical standard
  errors of estimators.}
\end{figure}
\end{center}

Though this paper contains no theoretical results describing the behavior of our
estimators for non-normal data, the numerical results in this example
suggest that some of the methods proposed here may be
successfully applied in broader circumstances.  The results in Table 1
for $\x_i \in \{\pm1\}$ binary show that all of the estimators
considered in this example are nearly unbiased and have standard
errors that are similar to the corresponding standard errors in the case where $\x_i \sim N(0,I)$.  Figure
4 contains histograms for the estimators $\tilde{\s}^2$,
$\tilde{\tau}^2$, and $\tilde{\tau}^2/\tilde{\s}^2$, with $\x_i \in
\{\pm 1\}$ binary.  Normal density plots
with mean 1 and variance $\psi_1^2$, $\psi_2^2$, and $\psi_0^2$ are
superimposed on the histograms; these are the normal
densities corresponding to the asymptotic distribution of the
estimators in the case where $\x_i \sim N(0,I)$ (see Corollaries 1-2
and Proposition 2 (ii)).  The histograms appear to match the densities
quite well.  
\begin{center}
\begin{figure}[h]
\includegraphics[width=\textwidth]{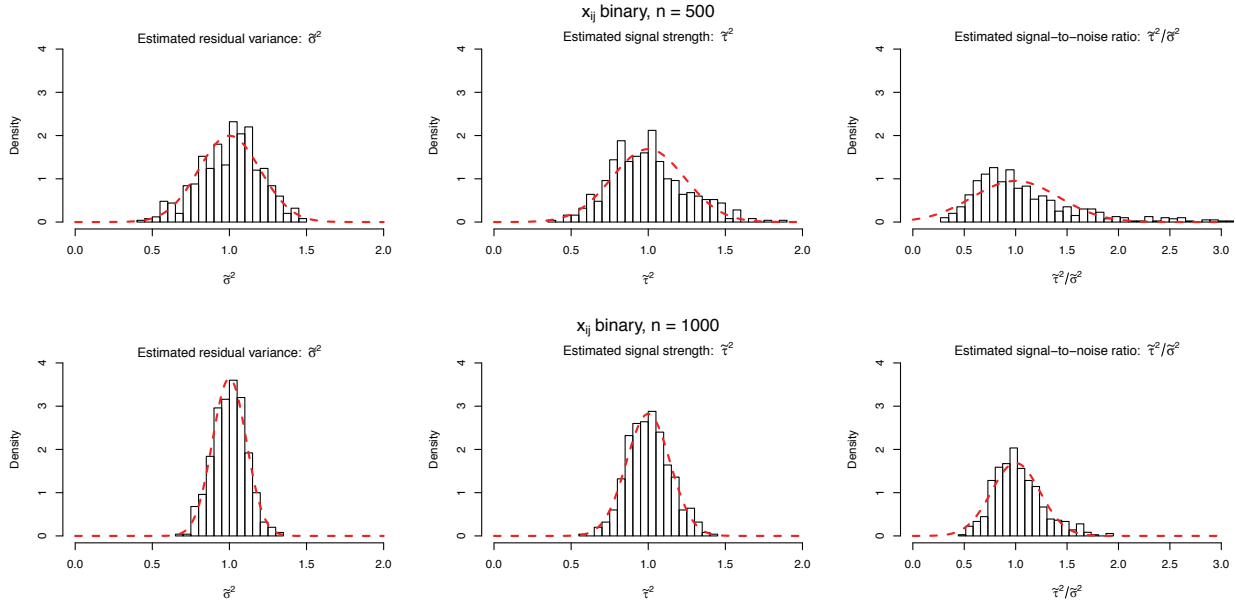}
\caption{Example  1 ($d = 1000$).  Histograms and normal density plots for the estimators
  $\tilde{\s}^2$, $\tilde{\tau}^2$, and
  $\tilde{\tau}^2/\tilde{\s}^2$, with $\x_i \in \{\pm 1\}$ binary.  Top row, $n = 500$; bottom row, $n
  = 1000$.  Superimposed normal density plots have mean 1
  and variance $\psi_1^2$, $\psi_2^2$, and $\psi_0^2$ for
  $\tilde{\s}^2$, $\tilde{\tau}^2$, and
  $\tilde{\tau}^2/\tilde{\s}^2$, respectively. }
\end{figure}
\end{center}
\subsection{Example 2}

When $d < n$, $\hat{\s}_0^2 = (n-d)^{-1}||\y - X\hat{\bb}_{ols}||^2$
is a widely used estimator for $\s^2$.  In Remark 2 following
Theorem 2, we noted that the variance of $\hat{\s}_0^2$ does not
depend on $\tau^2$, while the variance of $\hat{\s}^2(I)$ and the
other estimators for $\s^2$ proposed in this paper increases with
$\tau^2$.  On the other hand, as $d/n \uparrow 1$, the variance of
$\hat{\s}_0^2$ diverges, while that of $\hat{\s}^2(I)$ remains
bounded.  In this brief example, we took $\x_i \sim N(0,I)$, $\s^2 = \tau^2
= 1$, and $n = 500$, and investigated the numerical performance of
$\hat{\s}^2(I)$ and $\hat{\s}^2_0$ for various values of $d < n$.
Five hundred independent datasets were generated and the estimators
were computed for each dataset.  
Summary statistics are reported in Table 2.  

\begin{center}
\begin{table}[h]
\begin{tabular}{ll|r|r|r}
& Estimator & $d = 250$ & $d = 350$ & $d = 450$ \\ \hline
Mean & $\hat{\s}^2(I)$ & 0.9984 & 0.9986 & 0.9965 \\ \hline
& $\hat{\s}^2_0$ & 0.9979 & 1.0004 & 0.9902 \\ \hline \hline
Standard error & $\hat{\s}^2(I)$ & 0.1290 & 0.1389 & 0.1457 \\ \hline
& $\hat{\s}^2_0$ & 0.0901 & 0.1141 & 0.1947  
\end{tabular}
\caption{Example 2 ($n = 500$, $\s^2 = 1$).  Means and standard errors of
  estimators for $\s^2$, based on 500 independent datasets.  }
\end{table}
\end{center}

Table 2 indicates that in each setting, the estimators are nearly
unbiased: the means of the estimators are close to 1.  The empirical standard errors
of $\hat{\s}^2(I)$ and $\hat{\s}^2_0$ both increase with $d$;
however, the standard errors increase more rapidly for
$\hat{\s}^2_0$.  At $d = 250, 350$, the empirical standard
error of $\hat{\s}^2_0$ is smaller than that of $\hat{\s}^2(I)$; at $d
= 450$, the trend reverses and the empirical standard error of $\hat{\s}^2(I)$ is smaller
than that of $\hat{\s}^2_0$.  As $d$ becomes closer to
$n = 500$, the empirical standard error of $\hat{\s}^2(I)$ should
remain bounded, while that of $\hat{\s}^2_0$ should diverge to $\infty$. The
results reported in this example suggest that even when $d < n$, there
may be settings where the estimators proposed in this paper may be preferred to over other commonly used estimators for $\s^2$; for
instance, when $d < n$, but $d$ is very close to $n$.  

\subsection{Example 3}

\cite{sun2011scaled} proposed methods for estimating $\s^2$ in
high-dimensional linear models that are very effective when $\bb$ is
sparse.  These methods use modified versions of lasso
\citep{tibshirani1996regression} and MC+ \citep{zhang2010nearly},
(referred to as ``scaled lasso'' and ``scaled MC+,'' respectively)  to
simultaneously estimate $\s^2$ and $\bb$.  Let
$\hat{\s}_{\lasso}^2$ and $\hat{\s}_{\MCP}^2$ denote the scaled lasso
and scaled MC+ estimators for $\s^2$.  In this example, we
compared the performance of $\hat{\s}_{\lasso}^2$ and
$\hat{\s}_{\MCP}^2$ with some of the estimators for $\s^2$ proposed in
this paper, in settings where $\bb$ was both sparse and non-sparse.  

With $d = 3000$, the predictors in this example were generated according to $\x_i \sim
N(0,\S)$, where $\S = (\s_{ij})$ and $\s_{ij} = 0.5^{|i - j|}$.  We
fixed $\s^2 = 1$.   Sparse and non-sparse (dense) parameters $\bb \in \R^d$
were generated as follows.  First, to generate the sparse $\bb$, five
random multiples of 25 between $25$ and $d - 25 = 2975$ were
selected.  That is, we selected $k_1,...,k_5$ from
$\{25,50,75...,2975\}$ independently and uniformly at random.  Next, we took $\bb_0\in \R^d$ to be the vector with the 7-dimensional
sub-vector $(1,2,3,4,3,2,1)^T$ centered at the coordinates corresponding
to $k_1,...,k_5$ (so that the $k_j$-th entry of $\bb_0$ was 4, the
$(k_j \pm 1)$-th was 3, etc.); the remaining entries in
$\bb_0$ were set equal to 0.  We then set $\bb
=\{3/(\bb_0^T\S\bb_0)\}^{1/2}\bb_0$, so that $\tau_1^2 = \bb^T\S\bb =
3$.  Note that this sparse $\bb$ was generated only once; in other words, the same sparse
$\bb$ was use throughout the simulations in this example.  To generate
the dense $\bb$ used in this example, we followed the same procedure as for the sparse
$\bb$, except that in $\bb_0$, the 7-dimensional subvector $(1,2,3,4,3,2,1)^T$ was
centered at coordinates corresponding to {\em each} multiple of 25
between 25 and $2975$.  Notice that for the sparse $\bb$, we had
$||\bb||_0 = 7 \times 5 = 35$, where $||\bb||_0$ denotes the number of non-zero
coordinates in $\bb$,  and for the dense $\bb$ we had $||\bb||_0 =
7\times(d/25 - 1) =  833$; however, $\tau_1^2 = \bb^T\S\bb = 3$ was the
same for both the sparse and dense $\bb$.   In this simulation study, we considered datasets
with $n = 600$ and $n = 2400$ observations.  With sparse $\bb$ and $n
= 300$, the simulation settings in this example are very similar to those in Example 1 from
Section 4.1 of \citep{sun2011scaled}.  

Under each of the settings described above, we generated $100$ independent
datasets and, for each simulated dataset, we computed
$\hat{\s}^2_{\lasso}$, $\hat{\s}^2_{\MCP}$, 
$\hat{\s}^2(\hat{\S})$, $\hat{\s}^2(\S)$, and $\tilde{\s}^2$.  For the scaled lasso
and MC+ estimators, we used the shrinkage parameter $\l_0 =
\sqrt{\log(d)/n}$ (this value of $\l_0$ yielded the best performance in the
numerical examples in \citep{sun2011scaled}).  The scaled MC+
estimator requires specification of an additional parameter $\g$;
following \citep{sun2011scaled}, we
took $\g = 2/[1 - \max_{i,j} \{\X_i^T\X_j/(||\X_i||||\X_j||)\}]$,
where $\X_j$ denotes the $j$-th column of $X$.  The estimator
$\hat{\s}^2(\hat{\S})$ was introduced in Section 3.1 of this paper.
Here we take advantage of the AR(1) structure of $\S$ and set
$\hat{\S} = (\hat{\s}_{ij})$, where $\hat{\s}_{i,j} =
\hat{\a}^{|i-j|}$ and
\[
\hat{\a} = \frac{1}{n(d-1)}\sum_{i = 1}^n \sum_{j = 2}^d x_{ij}x_{i(j-1)}.
\]
We view the estimator $\hat{\s}^2(\S)$ as an ``oracle estimator,''
which utilizes full knowledge of actual covariance matrix $\S$;
this estimator should perform similarly to the estimator
$\hat{\s}^2(I)$ in settings where $\Cov(\x_i) = I$ and $\tau_1^2 =
3$ (see the discussion in Section 3.1).  Finally, the estimator
$\tilde{\s}^2$ is the ``unknown covariance'' estimator from Section
3.2.  Recall that our theoretical performance guarantees for
$\tilde{\s}^2$ (Proposition 2) require that $\left|\bb^T\S^k\bb -
||\bb||^2 \tr(\S^k)/d\right| \approx 0$, for $k = 1,2$.  In this
example, for the sparse $\bb$ we had
\begin{equation}\label{ex3a}
\frac{||\bb||^2}{d} \tr(\S) - \bb^T\S\bb = 
-1.7551 \ \mbox{ and }\  \frac{||\bb||^2}{d} \tr(\S^2) - \bb^T\S^2\bb 
= -5.5409
\end{equation}
(the corresponding quantities are essentially the same for the dense
$\bb$).  Summary statistics for the various estimators computed in this numerical study are reported in
Table 3.  

\begin{center}
\begin{table}[h]
\begin{minipage}[b]{0.4\linewidth}\centering
\begin{tabular}{ll|r|r}
\multicolumn{4}{c}{{Sparse $\bb$} } \\ 
&&  Mean & Std. Err. \\ \hline 
 $n = 600$ & $\hat{\s}^2_{\lasso}$ &  1.1117 & 0.0651 \\ \hline 
& $\hat{\s}^2_{\MCP}$ & 1.0477 & 0.0633 \\ \hline 
& $\hat{\s}^2(\hat{\S})$ & 0.9704 & 0.5049 \\ \hline
& $\hat{\s}^2(\S)$ & 0.9693 & 0.5021 \\ \hline 
& $\tilde{\s}^2$ & -0.6023 & 0.5182 \\ \hline \hline 
 $n = 2400$ & $\hat{\s}^2_{\lasso}$ & 1.0310 & 0.0295 \\ \hline 
& $\hat{\s}^2_{\MCP}$ & 1.0060 & 0.0293 \\ \hline 
& $\hat{\s}^2(\hat{\S})$ & 0.9808 & 0.1633 \\ \hline
& $\hat{\s}^2(\S)$ & 0.9809 & 0.1631\\ \hline 
& $\tilde{\s}^2$ & - 0.5827 & 0.2084
\end{tabular}
\end{minipage}\hspace{.25in}
\begin{minipage}[b]{0.4\linewidth} \centering
\begin{tabular}{ll|r|r}
\multicolumn{4}{c}{Dense $\bb$} \\ 
&&  Mean & Std. Err. \\ \hline 
 $n = 600$ & $\hat{\s}^2_{\lasso}$ & 3.2600 & 0.2070 \\ \hline 
& $\hat{\s}^2_{\MCP}$ & 3.1005 & 0.2107 \\ \hline 
& $\hat{\s}^2(\hat{\S})$ & 0.9820 & 0.5641 \\ \hline
& $\hat{\s}^2(\S)$ & 0.9835 & 0.5596 \\ \hline 
& $\tilde{\s}^2$ & -0.5747 & 0.5876 \\ \hline \hline 
 $n = 2400$ & $\hat{\s}^2_{\lasso}$ & 2.3232 & 0.0706 \\ \hline 
& $\hat{\s}^2_{\MCP}$ & 1.9997 & 0.0778 \\ \hline 
& $\hat{\s}^2(\hat{\S})$ & 1.0095 & 0.1538 \\ \hline
& $\hat{\s}^2(\S)$ & 1.0095 & 0.1537 \\ \hline 
& $\tilde{\s}^2$ & -0.5702 & 0.2228
\end{tabular}
\end{minipage}
\caption{Example 3 ($d = 3000$, $\s^2  =1$).  Means and standard errors of
  estimators for $\s^2$, based on 100 independent datasets.  Left table,
  sparse $\bb$; right table, dense $\bb$}
\end{table}
\end{center}

For sparse $\bb$, the results in Table 3 indicate that
$\hat{\s}^2_{\lasso}$, $\hat{\s}^2_{\MCP}$, $\hat{\s}^2(\hat{\S})$,
and $\hat{\s}^2(\S)$ are all nearly unbiased (recall that $\s^2 =1$ in
this example).  However, the empirical standard
errors for the scaled lasso and MC+ estimators are considerably smaller than the standard
errors for $\hat{\s}^2(\hat{\S})$ and $\hat{\s}^2(\S)$.  Note that in
this example, the
performance of $\hat{\s}^2(\hat{\S})$ is very similar to that of the oracle estimator
  $\hat{\s}^2(\S)$.  

The estimator
$\tilde{\s}^2$ is significantly biased in this example.  Indeed, the mean value of
$\tilde{\s}^2$ is negative, while $\s^2 > 0$.  The
poor performance of $\tilde{\s}^2$ in this example is not completely
unexpected, given that $\left|\bb^T\S^k\bb -
||\bb||^2 \tr(\S^k)/d\right|$ is substantially larger than 0 for $k =
1,2$ (see (\ref{ex3a})).
In fact, more can be said.  Using the approximation $\hat{m}_k \approx
m_k = \tr(\S^k)/d$, $k = 1,2$, one can check that 
\[
E(\tilde{\s}^2) \approx \s^2 + \tau_1^2 -\frac{m_1}{m_2}\tau_2^2.
\]   
Thus, the bias of $\tilde{\s}^2$ is approximately $\tau_1^2 - (m_1/m_2)\tau_2^2$.
In this example, $\tau_1^2 - (m_1/m_2)\tau_2^2 = -1.5700$ and 
\[
E(\tilde{\s}^2) \approx -0.5700.
\]
(this calculation is for the sparse $\bb$; the result is almost
exactly the same for the dense $\bb$).  Note the similarity between this approximation and the empirical means
of $\tilde{\s}^2$ in Table 3.  

For dense $\bb$, the performance of $\hat{\s}_{\lasso}^2$ and $\hat{\s}_{\MCP}^2$ breaks down,
while the performance of $\hat{\s}^2(\hat{\S})$, $\hat{\s}^2(\S)$, and
$\tilde{\s}^2$ remains
virtually unchanged, as compared to the sparse $\bb$ case.  When $n =
600$, the
empirical means of $\hat{\s}_{\lasso}^2$ and $\hat{\s}_{\MCP}^2$ are
both greater than 3; when $n = 2400$,  the
empirical means of $\hat{\s}_{\lasso}^2$ and $\hat{\s}_{\MCP}^2$ are
both  nearly greater than 2.  Both $\hat{\s}^2_{\lasso}$ and $\hat{\s}^2_{\MCP}$
depend on associated lasso and $\MCP$ estimators for $\bb$.  The performance break-down
of $\hat{\s}^2_{\lasso}$ and $\hat{\s}^2_{\MCP}$
when $\bb$ is dense is likely related to the fact that the corresponding
estimators for $\bb$ perform poorly when $\bb$ is dense and $d/n$ is large.  In Table 4, we report the
empirical mean squared error for the lasso and $\MCP$ estimators for $\bb$
that are associated with $\hat{\s}^2_{\lasso}$ and
$\hat{\s}^2_{\MCP}$; note that mean squared error is
substantially higher for estimating dense $\bb$.   

\begin{center}
\begin{table}[h]
\begin{minipage}[b]{0.35\linewidth}\centering
\begin{tabular}{l|r|r}
\multicolumn{3}{c}{{Sparse $\bb$} } \\ 
$n$ & lasso & MC+ \\ \hline
600 & 0.1888 & 0.3696 \\ \hline 
2400 & 0.0514 & 0.0894 
\end{tabular}
\end{minipage}
\begin{minipage}[b]{0.35\linewidth}\centering
\begin{tabular}{l|r|r}
\multicolumn{3}{c}{{Dense $\bb$} } \\ 
$n$ & lasso & MC+ \\ \hline
600 & 1.2176 & 1.2457 \\ \hline 
2400 &0.8961  & 0.9337  
\end{tabular}
\end{minipage}
\caption{Example 3 ($d = 3000$, $\s^2  =1$, $||\bb||^2 = 1.2449$).
  Empirical mean squared error $||\hat{\bb} - \bb||^2$ of the scaled
  lasso and MC+ estimators for $\bb$, based on 100 independent datasets.}
\end{table}
\end{center}  

Overall, the results of this simulation study suggest that estimators proposed in this paper
may be useful for estimating $\s^2$ in settings where $d/n$ is large
and little is know about sparsity in $\bb$.  However, we emphasize two important points:  (i) additional information about the
covariance matrix $\S$ may be required to obtain consistent estimators
for $\s^2$ (e.g. that $\S$ has AR(1) structure) and (ii) the 
estimators for $\s^2$ proposed in this paper may have larger standard
error than estimators derived from a reliable estimate of $\bb$.  

\section{Discussion}

In this paper, we proposed new estimators for $\s^2$, $\tau^2$, and
the signal-to-noise ratio $\tau^2/\s^2$ in high-dimensional linear
models.  These estimators are based on linear combinations of $T_1 = n^{-1}||\y||^2$ and
$T_2 = n^{-2}||X^T\y||^2$.  Working under the assumption that
$\Cov(\x_i) = I$, the key observation in deriving these
estimators was that $ET_1$, $ET_2$ form a pair of non-degenerate linear
combinations involving $\s^2$ and $\tau^2$.  In fact, as described in Section 2.1, 
unbiased estimators for $\s^2$ and $\tau^2$ may be derived from any pair
of statistics $T_1,T_2$ satisfying this property.  With 
$T_1=n^{-1}||\y||^2$ fixed, we presently discuss two 
alternatives for $T_2$, which may yield other estimators for $\s^2$,
$\tau^2$ in this manner.  These examples are not meant to be
exhaustive; rather, they are illustrative of this technique's
flexibility and raise some broader questions about estimating $\s^2$
and $\tau^2$ in high-dimensional linear models.

First, let
$U \in O(d)$ be a $d \times d$ Haar-distributed orthogonal matrix independent
of $(\y,X)$ and let $U_k$ denote the first $k$ columns of $U$, where
$1 \leq k \leq \min\{d,n\}$.  Then one may take $T_2 = n^{-1}E\left(||P_k\y||^2|\y,X\right)$, where $P_k =
\tilde{X}_k(\tilde{X}_k^T\tilde{X}_k)^{-1}\tilde{X}_k^T$ and
$\tilde{X}_k = XU_k$, so that $P_k$ is a random rank-$k$
projection. As a second alternative to $T_2 = n^{-2}||X^T\y||^2$, one could take $T_2 =
n^{-1}||X\hat{\bb}_{ridge}||^2$, where $\hat{\bb}_{ridge}$ is some
ridge regression estimator for $\bb$ \citep{hoerl1970ridge}.  One
aspect of these alternatives' potential appeal is that they might yield consistent
estimators for $\s^2$ and $\tau^2$ with smaller variance than the
estimators studied in this paper.  However, a theoretical
analysis of these estimators' properties may be somewhat involved.
Indeed, for $T_2 = n^{-1}E\left(||P_k\y||^2|\y,X\right)$, it is easy to calculate
$ET_2$ and find the corresponding unbiased estimators for $\s^2$ and $\tau^2$ using symmetry
arguments (provided $\Cov(\x_i) = I$), but computing the variance of
these estimators appears to be fairly challenging. If $T_2 = n^{-1}||X\hat{\bb}_{ridge}||^2$, then closed-form expressions for $ET_2$ and, consequently, for the associated unbiased estimators of $\s^2$, $\tau^2$
are generally not available; however, results from random matrix
theory suggest that simplified asymptotic analyses may be possible.
Note that in order to implement either of these alternatives to $T_2 =
n^{-2}||X^T\y||^2$, specification of an additional tuning parameter is required:  for $T_2 =
n^{-1}E\left(||P_k\y||^2|\y,X\right)$, the rank parameter $k$ must be specified; for $T_2 = n^{-1}||X\hat{\bb}_{ridge}||^2$, the ridge
shrinkage parameter (typically, a nonnegative constant denoted by
$\l$) must be specified.  

A number of questions are raised by the examples discussed in the
previous paragraph.  For instance, it is clear that
estimators for $\s^2$, $\tau^2$ derived using different statistics
$T_1$, $T_2$ may (or may not!) be more efficient than the estimators
$\hat{\s}^2$, $\hat{\tau}^2$ studied here; however, an exhaustive study of all
pairs $T_1,T_2$ aimed at identifying the optimal estimators for $\s^2$, $\tau^2$ is
likely impossible.  This suggests the need for a more unified approach
to studying efficiency and optimality for estimating $\s^2$ and
$\tau^2$ in high-dimensional linear models, which, given the ambiguity of
likelihood-based approaches noted in Section 1.3, may be challenging.
Additionally, while we have shown that the proposed approach to estimating $\s^2$ and $\tau^2$ based on linear
combinations of statistics $T_1$, $T_2$ is effective when $\Cov(\x_i)
= \S$, and that this approach may be successfully modified when
$\S$ satisfies additional conditions, it is unclear
whether a similar approach may be applied effectively when $\S$ is unknown
and arbitrary.  Studying different statistics $T_1$, $T_2$ may
provide additional insight into this problem, but other methodologies
may be required to handle more general $\S$.  

\section*{Appendix}
\subsection*{Proof of Theorem 2}

Theorem 2 is an immediate
consequence of the following lemma and its corollary.  

\setcounter{lemma}{0}
    \renewcommand{\thelemma}{A\arabic{lemma}}

\begin{lemma}
Suppose that $\S = I$. Then
\begin{eqnarray} \label{lemma1a}
\Var\left(\frac{1}{n}||\y||^2\right)\!\!\! & =& \!\!\! \frac{2}{n}(\s^2 + \tau^2)^2 \\ \nonumber
\Var\left(\frac{1}{n^2}||X^T\y||^2\right)\!\!\! & =& \!\!\! \frac{2}{n}\left[\left\{\left(\frac{d}{n}\right)^2 + \frac{d}{n} +
    \frac{2d}{n^2}\right\}\s^4 \right.\\ \label{lemma1b} && \  +
\left\{2\left(\frac{d}{n}\right)^2 + \frac{6d}{n} + 2 + \frac{10d}{n^2} +
    \frac{10}{n} + \frac{12}{n^2}\right\}\s^2\tau^2 \\ \nonumber
&& \  \left. + \left\{\left(\frac{d}{n}\right)^2 + \frac{5d}{n} + 4 +
  \frac{8d}{n^2} + \frac{15}{n} + \frac{15}{n^2}\right\}\tau^4\right]
\\ \label{lemma1c}
\Cov\left(\frac{1}{n}||\y||^2,\frac{1}{n^2}||X^T\y||^2\right) \!\!\!& = &\!\!\!
\frac{2}{n}\left\{\frac{d}{n}\s^4 + \left(\frac{2d}{n} + 2 +
    \frac{3}{n}\right)\s^2\tau^2+ \left(\frac{d}{n} + 2 +
    \frac{3}{n}\right)\tau^4\right\}.
\end{eqnarray}
\end{lemma} 

\begin{proof}
Equation (\ref{lemma1a}) is obvious because $||\y||^2 \sim
(\s^2 + \tau^2)\chi^2_n$.  To prove (\ref{lemma1b}), we condition on
$X$ and use properties of expectations involving quadratic forms and
normal random vectors to obtain
\begin{eqnarray*}
\Var(||X^T\y||^2) & = & E\left\{\Var(||X^T\y||^2|
  X)\right\} + \Var\left\{E(||X^T\y||^2|X)\right\} \\
& = & 2\s^4E\tr\left\{(X^TX)^2\right\} +
  4\s^2E\left\{\bb^T(X^TX)^3\bb\right\}  \\
&& \ + \Var\left\{\s^2\tr(X^TX) + \bb^T(X^TX)^2\bb\right\} \\
& = & 2\s^4E\tr\left\{(X^TX)^2\right\} +
  4\s^2E\left\{\bb^T(X^TX)^3\bb\right\} + \s^4E\left\{\tr(X^TX)\right\}^2  \\
&& \ +
2\s^2E\left\{\tr(X^TX)\bb^T(X^TX)^2\bb\right\} + E\left\{\bb^T(X^TX)^2\bb\right\}^2 -
\s^4\left\{E\tr(X^TX)\right\}^2 \\
&& \ - 2\s^2E\tr(X^TX)E\left\{\bb^T(X^TX)^2\bb\right\} - \left[E\left\{\bb^T(X^TX)^2\bb\right\}\right]^2.
\end{eqnarray*}
Given this expression for $\Var(||X^T\y||^2)$, (\ref{lemma1b}) follows
from Proposition S1 in the Supplemental Text.  Equation
(\ref{lemma1c}) is proved similarly: we have
\begin{eqnarray*}
\Cov(||\y||^2,||X^T\y||^2) & =&
E\left\{\Cov(||\y||^2,||X^T\y||^2|X)\right\}+
\Cov\left\{E(||\y||^2|X),E(||X^T\y||^2|X)\right\} \\
& = & 2\s^4E\tr(X^TX) + 4\s^2E\left\{\bb^T(X^TX)^2\bb\right\} \\
&& \ + \Cov\left\{\bb^TX^TX\bb,\s^2\tr(X^TX) +
  \bb^T(X^TX)^2\bb\right\}\\
& = & 2\s^4E\tr(X^TX) + 4\s^2E\left\{\bb^T(X^TX)^2\bb\right\} + \s^2E\left\{\tr(X^TX)\bb^TX^TX\bb\right\} \\
&& \ + E\left\{\bb^TX^TX\bb\bb^T(X^TX)^2\bb\right\} - \s^2E\left(\bb^TX^TX\bb\right)E\tr(X^TX) \\
&& \ -  E\left(\bb^TX^TX\bb\right)E\left\{\bb^T(X^TX)^2\bb\right\}
\end{eqnarray*}
and (\ref{lemma1c}) follows from Proposition S1 in the Supplemental
Text.  
\end{proof}

\setcounter{cor}{0}
    \renewcommand{\thecor}{A\arabic{cor}}

\begin{cor}
Under the conditions of Lemma 1,
\begin{eqnarray*}
\Var(\hat{\s}^2) & =& \frac{2n}{(n+1)^2}\left\{\left(\frac{d}{n} + 1 +
    \frac{2d}{n^2}  + \frac{2}{n} +
\frac{1}{n^2}\right)\s^4  + \left(\frac{2d}{n} + \frac{4d}{n^2} + \frac{4}{n} + \frac{8}{n^2}\right)\s^2\tau^2 \right.\\
&& \qquad \left. + \left(\frac{d}{n} + 1 + \frac{2d}{n^2} + \frac{7}{n} +
    \frac{10}{n^2}\right)\tau^4\right\}\\
\Var(\hat{\tau}^2) & = & \frac{2n}{(n+1)^2}\left\{\left(\frac{d}{n} +
    \frac{2d}{n^2}\right)\s^4 
+ \left(\frac{2d}{n} + 2 + \frac{4d}{n^2} + \frac{10}{n} +
  \frac{12}{n^2}\right) \s^2\tau^2 \right. \\
&& \qquad  \left. + \left(\frac{d}{n} + 4 + \frac{2d}{n^2} + \frac{15}{n} +
    \frac{15}{n^2}\right)\tau^4\right\} \\
\Cov(\hat{\s}^2,\hat{\tau}^2) & = &
-\frac{2n}{(n+1)^2}\left\{\left(\frac{d}{n} +
    \frac{2d}{n^2}\right)\s^4+ \left(\frac{2d}{n} + \frac{4d}{n^2} + \frac{5}{n} +
  \frac{9}{n^2}\right)\s^2\tau^2 \right. \\
&& \qquad \left. + \left(\frac{d}{n} + 2 + \frac{2d}{n^2} + \frac{10}{n} +
    \frac{12}{n^2}\right)\tau^4\right\}.
\end{eqnarray*}
\end{cor}

\begin{proof}
Corollary 1 follows from Lemma 1 and the fact that
\[
\left(\begin{array}{c} \hat{\s}^2 \\ \hat{\tau}^2 \end{array} \right) =
\left(\begin{array}{cc} \frac{d + n + 1}{n+1} & -\frac{n}{n+1} \\
    -\frac{d}{n+1} & \frac{n}{n+1} \end{array}\right)\left(\begin{array}{c}
    n^{-1}||\y||^2 \\ n^{-2}||X^T\y||^2 \end{array}\right).
\]
\end{proof}

\subsection*{Proof of Theorem 3}

Theorem 3 is a direct application of Theorem 2.2 from
\citep{chatterjee2009fluctuations}, which is stated here for ease of reference.  

\setcounter{thm}{0}
    \renewcommand{\thethm}{A\arabic{thm}}

\begin{thm} {\em [Theorem 2.2, \citep{chatterjee2009fluctuations}]}
  Let $\vv = (v_1,...,v_m)^T \sim N(0,\Psi)$.  Suppose that $g \in
  C^2(\R^m)$ and let $\nabla g$ and $\nabla^2 g$ denote the gradient
  and the Hessian of $g$, respectively.  Let
\begin{eqnarray*}
\kappa_1 & = & \left\{E||\nabla g(\vv)||^4\right\}^{1/4} \\
\kappa_2 & = & \left\{E||\nabla^2 g(\vv)||^4\right\}^{1/4},
\end{eqnarray*}
where $||\nabla^2g(\vv)||$ is the operator norm of $\nabla^2g(\vv)$.  
Suppose that $Eg(\vv)^4 < \infty$ and let $\psi^2 = \Var\{g(\vv)\}$.
Let $w$ be a normal random variable having the same mean and variance as
$g(\vv)$.  Then
\begin{equation}\label{thm3a}
d_{TV}\{g(\vv),w\} \leq \frac{2\sqrt{5}||\Psi||^{3/2} \kappa_1\kappa_2}{\psi^2}.
\end{equation}

\setcounter{rk}{0}

\begin{rk}
Chatterjee's Theorem 2.2 does not actually require Gaussian $\vv$.
However, for non-Gaussian $\vv$, an additional term appears
in the bound (\ref{thm3a}), which is not sufficiently small for our
purposes.  Furthermore, the class of distributions covered by the full version of
Chatterjee's Theorem 2.2 is not all-encompassing: $v_i$ must be a
$C^2$-function of a normal random variable.  
\end{rk}
\end{thm}

To prove Theorem 3, we apply Theorem A1 with $\vv = (X,\ee) \in \R^{(d+1)n}$.  Let
$h \in C^2(\R^2)$ and let 
\[
g(X,\ee) = h(\T),
\]
where $\T = \T(X,\ee) = (n^{-1}||\y||^2,n^{-2}||X^T\y||^2)^T$.
First, we bound the quantities $\kappa_1$, $\kappa_2$ in Theorem A1.
In order to bound $\kappa_1$, we 
compute the gradient of $g$.  Let $h_1$, $h_2$ denote the
partial derivatives of $h$ with respect to the first and second
variables, respectively.  Then
\[
\frac{\partial g}{\partial x_{ij}} (X,\ee) =
h_1(\T)\frac{\partial}{\partial x_{ij}}
\frac{1}{n}||\y||^2 +
h_2(\T)\frac{\partial}{\partial x_{ij}} \frac{1}{n^2}||X^T\y||^2
\]
for $i = 1,...,n$, $j = 1,...,d$.  Let $E_{ij}$ denote the $n\times d$ matrix with $i'j'$-entry $\d_{ii'}\d_{jj'}$
($\d_{ii'} = 1$ if $i = i'$ and $0$ otherwise). Since 
\[
\frac{\partial}{\partial x_{ij}} ||\y||^2  = 2\bb^TE_{ij}^T\y
\]
and 
\[
\frac{\partial}{\partial x_{ij}} ||X^T\y||^2 = 2\y^TE_{ij}X^T\y + 2\bb^TE_{ij}^TXX^T\y,
\]
it follows that
\begin{equation}\label{thm1a}
\frac{\partial g}{\partial x_{ij}} (X,\ee)  = 
2h_1(\T)\left(\frac{1}{n}\bb^TE_{ij}^T\y\right) +
2h_2(\T)\left(\frac{1}{n^2}\y^TE_{ij}X^T\y + \frac{1}{n^2}\bb^TE_{ij}^TXX^T\y\right).
\end{equation}
For $1 \leq k \leq n$, the partial derivative of $g$ with respect to $\e_k$ is given by
\begin{eqnarray} \nonumber
\frac{\partial g}{\partial \e_k}(X,\ee) & = &
h_1(\T)\frac{\partial}{\partial \e_k}
\frac{1}{n}||\y||^2 +
h_2(\T)\frac{\partial}{\partial \e_k}
\frac{1}{n^2}||X^T\y||^2 \\ \label{thm1b}
& = & 2h_1(\T)
\left(\frac{1}{n}\mathbf{e}_k^T\y\right) +
2h_2(\T) \left(\frac{1}{n^2}\mathbf{e}_k^TXX^T\y\right),
\end{eqnarray}
where $\mathbf{e}_k \in \R^n$ is the $k$-th standard basis vector in
$\R^n$ (i.e. the $k'$-th entry of $\mathbf{e}_k$
is $\d_{kk'}$) and we have used the facts
\begin{eqnarray*}
\frac{\partial}{\partial \e_k} ||\y||^2 & = & 2\mathbf{e}_k^T\y \\
\frac{\partial}{\partial \e_k} ||X^T\y||^2 & = & 2\mathbf{e}_kXX^T\y.
\end{eqnarray*}  
Now recall that $\kappa_1 = \left(E||\nabla
    g(X,\ee)||^4\right)^{1/4}$.  Equations (\ref{thm1a})-(\ref{thm1b}) and
the elementary inequality 
\begin{equation}\label{thm1c}
(a + b)^2 \leq 2a^2 + 2b^2, \ \ a,b \in \R,
\end{equation} 
imply that
\begin{eqnarray*}
||\nabla g(X,\ee)||^2 & = & \sum_{i = 1}^n \sum_{j = 1}^d
\left\{\frac{\partial}{\partial x_{ij}} g(X,\ee)\right\}^2 + \sum_{k =
  1}^n \left\{\frac{\partial}{\partial \e_k} g(X,\ee)\right\}^2 \\
& \leq & 8h_1(\T)^2\sum_{i = 1}^n \sum_{j = 1}^d
\left(\frac{1}{n}\bb^TE_{ij}^T\y\right)^2 \\
&& \ + 16h_2(\T)^2\sum_{i = 1}^n \sum_{j = 1}^d
\left\{\left(\frac{1}{n^2}\y^TE_{ij}X^T\y\right)^2 +
  \left(\frac{1}{n^2}\bb^TE_{ij}^TXX^T\y\right)^2\right\} \\
&& \ +  8h_1(\T)^2\sum_{k = 1}^n
\left(\frac{1}{n}\mathbf{e}_k^T\y\right)^2 + 8h_2(\T)^2\sum_{k = 1}^n
\left(\frac{1}{n^2}\mathbf{e}_k^TXX^T\y\right)^2 \\
& = & \frac{8}{n^2}(\tau^2+1)h_1(\T)^2||\y||^2 +
\frac{16}{n^4}h_2(\T)^2\left\{||\y||^2||X^T\y||^2 +
  \left(\tau^2 + \frac{1}{2}\right) ||XX^T\y||^2\right\}.
\end{eqnarray*}
Let $\l_1 = ||n^{-1}X^TX||$ be the largest eigenvalue of
$n^{-1}X^TX$.  Applying the triangle inequality and (\ref{thm1c}) yields
\begin{eqnarray*} 
||\nabla g(X,\ee)||^2 & \leq & \frac{16}{n^2}(\tau^2+1)h_1(\T)^2\left(||X^TX||\tau^2 + ||\ee||^2\right) \\
&& \ +
\frac{128}{n^4}h_2(\T)^2||X^TX||\left(||X^TX||^2\tau^4
  + ||\ee||^4\right) \\
&& \ + \frac{32}{n^4}h_2(\T)^2
||X^TX||^2\left(\tau^2 + \frac{1}{2}\right)\left(||X^TX||\tau^2 + ||\ee||^2\right) \\
& \leq & \frac{16}{n}||\nabla
h(\T)||^2\Bigg\{8\l_1\left(\frac{1}{n}||\ee||^2\right)^2+
 \left(2\l_1^2 + 1\right) \frac{1}{n}||\ee||^2\tau^2 +\left(10\l_1^3 +
   \l_1 \right) \tau^4 \\
&& \  + (\l_1^2 + 1) \frac{1}{n}||\ee||^2 + (\l_1^3 + \l_1)\tau^2
\Bigg\} \\
& \leq & \frac{264}{n}||\nabla h(\T)||^2(\l_1 +
1)^3\left\{\frac{1}{n}||\ee||^2\left(\frac{1}{n}||\ee||^2 + 1\right) +
\tau^2(\tau^2 + 1)\right\}.
\end{eqnarray*}
Thus,
\begin{eqnarray}\nonumber
\kappa_1 & = & \left(E||\nabla g(X,\ee)||^4\right)^{1/4} \\ \nonumber
& \leq & \sqrt{\frac{264}{n}}\left(E \left[||\nabla
    h(\T)||^4 (\l_1 + 1)^6 \left\{\frac{1}{n}||\ee||^2\left(\frac{1}{n}||\ee||^2 + 1\right) +
\tau^2(\tau^2 + 1)\right\}^2\right]\right)^{1/4} \\ \label{thm1d}
& = & O\left[\frac{1}{\sqrt{n}}\left\{\g_4^{1/4} + \g_2^{1/4} + \g_0^{1/4}\tau(\tau
    + 1)\right\}\right],
\end{eqnarray}
where 
\[
\g_k = E\left[||\nabla h(\T)||^4(\l_1 + 1)^6\left(\frac{1}{n}||\ee||^2\right)^k\right].
\]

To bound $\kappa_2 = \left\{E||\nabla^2g(X,\ee)||^4\right\}^{1/4}$, we
bound the operator norm of the Hessian $||\nabla^2 g(X,\ee)||$.  Let 
\[
\mathcal{U} = \left\{\tilde{U} = (\u \  U);  \ \u =
    (u_1,...,u_n)^T \in \R^n,  \ U = (u_{ij})_{1 \leq i \leq n, \ 1
      \leq j \leq d}, \ \sum_{k = 1}^n u_k^2 + \sum_{i = 1}^n \sum_{j
   = 1}^d u_{ij}^2 = 1\right\}
\]
be the collection of partitioned $n \times (d+1)$ matrices with
Frobenius norm equal to one.  For $\tilde{U} = (\u \ U) \in \mathcal{U}$, define
the differential operator 
\[
D_{\tilde{U}} = \sum_{i = 1}^n \sum_{j = 1}^d
u_{ij}\frac{\partial}{\partial x_{ij}} + \sum_{k = 1}^n
u_k \frac{\partial}{\partial \e_k}.
\]
Then
\begin{eqnarray}\nonumber
||\nabla^2g(X,\ee)|| & = & \sup_{\tilde{U} \in \mathcal{U}}
D_{\tilde{U}}^2g(X,\ee) \\ \nonumber 
& = & \sup_{\tilde{U} \in \mathcal{U}}
\left\{\nabla h(\T)^TD_{\tilde{U}}^2\T(X,\ee) +
  \left\{D_{\tilde{U}}\T(X,\ee)\right\}^T\nabla^2h(\T)
  D_{\tilde{U}}\T(X,\ee)\right\} \\ \label{thm1e}
& \leq &  \sup_{\tilde{U} \in \mathcal{U}}
\left\{\left|\left|\nabla
      h(\T)\right|\right|\left|\left|D_{\tilde{U}}^2\T(X,\ee)\right|\right|
  + \left|\left|\nabla^2h(\T)\right|\right|\left|\left| D_{\tilde{U}}\T(X,\ee)\right|\right|^2\right\}.
\end{eqnarray}
From our previous calculations,
\begin{eqnarray*}
D_{\tilde{U}}\T(X,\ee) & = &  \sum_{i = 1}^n \sum_{j = 1}^d
u_{ij} \frac{\partial}{\partial x_{ij}} \left(\begin{array}{c} \frac{1}{n}||\y||^2 \\
    \frac{1}{n^2}||X^T\y||^2\end{array}\right) + \sum_{k = 1}^n u_k
\frac{\partial}{\partial \e_k}\left(\begin{array}{c} \frac{1}{n}||\y||^2 \\
    \frac{1}{n^2}||X^T\y||^2\end{array}\right) \\ 
& = & \sum_{i = 1}^n \sum_{j = 1}^d
u_{ij} \left(\begin{array}{c} \frac{2}{n}\bb^TE_{ij}^T\y \\
    \frac{2}{n^2}\y^TE_{ij}X^T\y + \frac{2}{n^2
    }\bb^TE_{ij}^TXX^T\y\end{array}\right) + \sum_{k = 1}^n u_k
\left(\begin{array}{c} \frac{2}{n}\mathbf{e}_k^T\y \\
    \frac{2}{n^2}\mathbf{e}_k^TXX^T\y \end{array}\right) \\
& = & \left(\begin{array}{c} \frac{2}{n}\bb^TU^T\y + \frac{2}{n}\u^T\y
    \\ \frac{2}{n^2} \y^TUX^T\y + \frac{2}{n^2}\bb^TU^TXX^T\y +
    \frac{2}{n^2}\u^TXX^T\y \end{array}\right).
\end{eqnarray*}
To compute $D^2_{\tilde{U}}\T(X,\ee)$, we need the second order
partial derivatives of $||\y||^2$ and $||X^T\y||^2$; these are given
below:
\begin{eqnarray*}
 \frac{\partial^2}{\partial
  x_{i'j'} \partial x_{ij}} ||\y||^2 & = &  2\bb^TE_{ij}^TE_{i'j'}\bb \\
\frac{\partial^2}{\partial \e_k \partial x_{ij}} ||\y||^2 & = &
2\bb^TE_{ij}^T\mathbf{e}_k \\
\frac{\partial^2}{\partial \e_{k'}\partial \e_k} ||\y||^2 & = &
2\mathbf{e}_k^T\mathbf{e}_{k'} 
\end{eqnarray*}
and
\begin{eqnarray*}
\frac{\partial^2}{ \partial x_{i'j'} \partial x_{ij}} ||X^T\y||^2 
& = & 2\bb^TE_{i'j'}^TE_{ij}X^T\y + 2\bb^TE_{ij}^TE_{i'j'}X^T\y +
2\y^TE_{ij}E_{i'j'}^T\y \\
&& \ + 2\y^TE_{ij}X^TE_{i'j'}\bb + 2\bb^TE_{ij}^TXE_{i'j'}^T\y +
2\bb^TE_{ij}^TXX^TE_{i'j'}\bb, \\
\frac{\partial^2}{\partial \e_k \partial x_{ij}} ||X^T\y||^2 & = &
2\mathbf{e}_k^TE_{ij}X^T\y + 2\y^TE_{ij}X^T\mathbf{e}_k +
2\bb^TE_{ij}^TXX^T\mathbf{e}_k \\
\frac{\partial^2}{\partial \e_{k'} \partial \e_k} ||X^T\y||^2 & = & 2\mathbf{e}_k^TXX^T\mathbf{e}_{k'},
\end{eqnarray*}
for $1 \leq i,k \leq d$ and $1 \leq j \leq d$.  It follows that the
entries of $D_{\tilde{U}}^2\T(X,\ee)$ are
\[
\frac{1}{n}D_{\tilde{U}}^2||\y||^2 = \frac{2}{n}\bb^TU^TU\bb +
\frac{4}{n}\bb^TU^T\u + \frac{2}{n}||\u||^2 
\]
and
\begin{eqnarray*}
\frac{1}{n^2}D_{\tilde{U}}^2||X^T\y||^2 & = &
\frac{2}{n^2}\y^TUU^T\y  + \frac{4}{n^2}\bb^TU^TUX^T\y  +
\frac{4}{n^2}\bb^TU^TXU^T\y +
\frac{2}{n^2} \bb^TU^TXX^TU\bb \\
&& \ + \frac{4}{n^2} \u^TUX^T\y + \frac{4}{n^2}\y^TUX^T\u +
\frac{4}{n^2}\bb^TU^TXX^T\u + \frac{2}{n^2} \u^TXX^T\u.
\end{eqnarray*}
We conclude that
\begin{eqnarray}\nonumber
\left|\left|D_{\tilde{U}}\T(X,\ee)\right|\right|^2 & = &
\frac{4}{n^2}\left(\bb^TU^T\y + \u^T\y\right)^2 +
\frac{4}{n^4}\left(\y^TUX^T\y + \bb^TU^TXX^T\y + \u^TXX^T\y\right)^2
\\ \nonumber
& \leq & \frac{8}{n^2}(\tau^2 + 1)||\y||^2 +
\frac{12}{n^4}||X^TX||\left(||\y||^2 + ||X^TX||\tau^2 +
  ||X^TX||\right)||\y||^2 \\ \nonumber
& \leq & \frac{16}{n}(\tau^2 + 1)\left(\l_1\tau^2 +
  \frac{1}{n}||\ee||^2\right) \\ \nonumber 
&& \ + \frac{168}{n}\l_1\left\{\l_1^2\tau^2(\tau^2 + 1) + \frac{1}{n}||\ee||^2\left(\l_1 +
    \frac{1}{n}||\ee||^2\right)\right\} \\ \label{thm1f} 
&= & O\left[\frac{1}{n}\left\{(\l_1^3 + \l_1)\tau^2(\tau^2 + 1) + \frac{1}{n}||\ee||^2\left(\l_1 +
    \frac{1}{n}||\ee||^2\right)\right\}\right]
\end{eqnarray}
and
\begin{eqnarray}\nonumber
||D_{\tilde{U}}^2\T(X,\ee)|| & \leq & \frac{2}{n}(\tau + 1)^2 +
\frac{2}{n^2}\left\{||\y||^2 + 4||X||(\tau + 1)||\y|| +
  ||X^TX||(\tau + 1)^2\right\} \\ \label{thm1g}
& = & O\left\{\frac{1}{n}\left(\l_1(\tau + 1)^2 + \frac{1}{n}||\ee||^2\right)\right\}.
\end{eqnarray}
Combining (\ref{thm1e})-(\ref{thm1g}), we obtain
\begin{equation} \label{thm1h}
\kappa_2 = \left(E||\nabla^2g(X,\ee)||^4\right)^{1/4} = O\left[\frac{1}{n}\left\{\eta_8^{1/4} +
    \eta_4^{1/4} + \eta_0^{1/4}\tau^2(\tau^2 + 1) + \g_4^{1/4} +
    \g_0^{1/4}(\tau^2 + 1)\right\}\right],
\end{equation}
where
\[
\eta_k = E\left[||\nabla^2h(\T)||^4(\l_1 + 1)^{12}\left(\frac{1}{n}||\ee||^2\right)^k\right].
\]
Appealing to Theorem A1, the bounds (\ref{thm1d}) and (\ref{thm1h})
imply 
\[
d_{TV}\left\{g(X,\ee),w\right\} = O\left(\frac{\xi\nu}{ n^{3/2}\psi^2}\right),
\]
where
\[
\xi = \xi(\s^2,\tau^2,\S,d,n) = \g_4^{1/4} + \g_2^{1/4} + \g_0^{1/4}\tau(\tau + 1)
\]
and
\[
\nu = \nu(\s^2,\tau^2,\S,d,n) = \eta_8^{1/4} +
    \eta_4^{1/4} + \eta_0^{1/4}\tau^2(\tau^2 + 1) + \g_4^{1/4} +
    \g_0^{1/4}(\tau^2 + 1).
\]
This completes the proof of Theorem 3. 

\subsection*{Proof outline for Proposition 2}

Let
\begin{eqnarray*}
\tilde{\s}^2(\hat{\m}) & = & \tilde{\s}^2 \ \ = \ \ \left\{1 +
\frac{d \hat{m}_1^2}{(n+1)\hat{m}_2}\right\}\frac{1}{n}||\y||^2 -
\frac{\hat{m}_1}{n(n+1)\hat{m}_2}||X^T\y||^2 \\
\tilde{\tau}^2(\hat{\m}) & = &\tilde{\tau}^2 \ \ = \ \ -\frac{d\hat{m}_1^2}{n(n+1)\hat{m}_2}||\y||^2 + \frac{\hat{m}_1}{n(n
  + 1)\hat{m}_2}||X^T\y||^2,
\end{eqnarray*}
where $\hat{\m} = (\hat{m}_1,\hat{m}_2)^T$.  With $\m = (m_1,m_2)^T =
(d^{-1}\tr(\S),d^{-1}\tr(\S^2))^T$, consider the estimators
$\tilde{\s}^2(\m)$ and $\tilde{\tau}^2(\m)$.  Under the conditions of
Proposition 2, Proposition S1 from the
Supplemental Text implies that $E(\hat{m}_k - m_k)^2 =
O(n^{-2})$, $k = 1,2$; furthermore, existing results on the eigenvalues of Wishart
matrices imply that $E\hat{m}_2^{-(2 + r)} =
O(1)$ for $r > 0$ sufficiently small (see,
for example, the Appendix of \citep{dicker2012dense}; this is where
the conditions that $|n-d| > 9$ and $d/n$ is bounded away from 1 are required).   These facts
can be combined to obtain
\begin{equation}\label{prop2a}
E\left\{\tilde{\s}^2(\hat{\m})  - \tilde{\s}^2(\m)\right\}^2 =
O\left(\frac{1}{n^2}\right) \mbox{ and } E\left\{\tilde{\tau}^2(\hat{\m}) -
\tilde{\tau}^2(\m)\right\}^2 = O\left(\frac{1}{n^2}\right).
\end{equation}
Additionally, it can be shown that
\begin{equation} \label{prop2b}
E \{\tilde{\s}^2(\m)\} = \s^2 + O(\tilde{\Delta}_2) \mbox{ and }
E\{\tilde{\tau}^2(\m)\} = \tau^2 + O(\tilde{\Delta}_2)
\end{equation}
and
\begin{equation} \label{prop2c}
\Var\{\tilde{\s}^2(\m)\} = \frac{\tilde{\psi}_1^2}{n} +
O\left(\frac{1 + n\tilde{\Delta}_3}{n^2}\right) \mbox{ and }
\Var\{\tilde{\tau}^2(\m)\} = \frac{\tilde{\psi}_2^2}{n} +
O\left(\frac{1 + n\tilde{\Delta}_3}{n^2}\right),
\end{equation}
where Proposition S1 in the Supplemental Text and the variance/covariance
decompositions in the proof of Lemma A1 are useful for proving
(\ref{prop2c}).  Part (i) of Proposition 2 (consistency) follows from
(\ref{prop2a})-(\ref{prop2c}). Part (ii) of Proposition 2 (asymptotic normality) also follows from
(\ref{prop2a})-(\ref{prop2c}), upon noticing that Theorem 3 may be
applied to $\tilde{\s}^2(\m)$ and $\tilde{\tau}^2(\m)$, as in
Corollary 1.  Asymptotic normality for $\tilde{\tau}^2/\tilde{\s}^2$
follows from the delta method.


\section*{Supplemental text: Moment calculations for the Wishart distribution} 

Suppose that $X = (\x_1,...,\x_n)^T$ is an $n \times d$ matrix with
iid rows $\x_1,...,\x_n \sim N(0,\S)$ and that $\S$ is a $d \times d$
positive definite matrix.  Then $W = X^TX$ is a
$\mathrm{Wishart}(n,\S)$ random matrix.  Let $\bb \in \R^d$.  In this
Supplemental Text we
provide formulas for various moments involving $W$ that are used
in the paper.
\citet{letac2004all} and \citet{graczyk2005hyperoctahedral} provide
techniques for computing all such moments.  These techniques are
utilized here.  

\subsection*{The symmetric group and a formula for a class of moments involving $W$}
 Let $S_k$ denote the
symmetric group on $k$ elements.  Then each permutation $\pi \in S_k$
can be uniquely as a product of disjoint cycles $\pi = C_1\cdots C_{m(\pi)}$, where
$C_j = (c_{1j} \cdots c_{k_jj})$, $k_1 + \cdots + k_{m(\pi)} = k$, and
all of the $c_{ij} \in \{1,...,k\}$ are distinct.  

Let $H_1,...,H_k$ be $d \times d$ symmetric matrices and define
the polynomial
\[
r_{\pi}(\S)(H_1,...,H_k) = \prod_{j = 1}^{m(\pi)}\tr\left(\prod_{i =
    1}^{k_j} \S H_{c_{ij}}\right).
\]
Theorem 1 in \cite{letac2004all} and Proposition 1 in
\cite{graczyk2005hyperoctahedral} give the following formula:
\begin{equation}\label{letac}
E\left\{\tr(WH_1)\cdots \tr(WH_k)\right\} = \sum_{\pi \in S_k} 2^{k-m(\pi)}n^{m(\pi)}r_\pi(\S)(H_1,...,H_k).
\end{equation}
This is our main tool for deriving the explicit formulas in the next
section. 

\subsection*{Explicit moment formulas used in the paper}

For non-negative integers $k$, define $\tau_k^2 = \bb^T\S^k\bb$ and $m_k
= d^{-1}\tr(\S^k)$. 

\setcounter{prop}{0}
    \renewcommand{\theprop}{S\arabic{prop}}

\begin{prop}
We have
\begin{eqnarray}
E\tr(W) & =& dnm_1 \label{prop1a} \\ 
E\tr(W)^2 & = & d^2n^2m_1^2 +2dnm_2\label{prop1b} \\ 
E\tr(W^2) & = & d^2nm_1^2 + dn(n+1)m_2 \label{prop1c} \\
E\bb^TW\bb & = & n\tau_1^2\label{prop1f} \\
E\bb^TW^2\bb & = & dnm_1\tau_1^2 + n(n+1)\tau_2^2 \label{prop1h} \\
E\left\{\tr(W)\bb^TW\bb\right\} & = & dn^2m_1\tau_1^2 + 2n\tau_2^2 \label{prop1i} \\
E\left\{\tr(W)\bb^TW^2\bb\right\} & = & d^2n^2m_1^2\tau_1^2 + dn(n^2 + n +
2)m_1\tau_2^2 \nonumber \\
&& \ \ + 2dnm_2\tau_1^2 + 4n(n+1)\tau_3^2 \label{prop1j} \\
E(\bb^TW\bb\bb^TW^2\bb) & = & dn(n+2)m_1\tau_1^4 + n(n+2)(n+3)\tau_1^2\tau_2^2\label{prop1k} \\
E\bb^TW^3\bb & = &  d^2nm_1^2\tau_1^2 + 2dn(n + 1)m_1\tau_2^2 \nonumber \\
&& \ \ + dn(n+1) m_2\tau_1^2 + n(n^2 +
3n + 4)\tau_3^2 \label{prop1l} \\
E(\bb^TW^2\bb)^2 & = & d^2n(n+2)m_1^2\tau_1^4 + 2dn(n+2)(n+3)m_1\tau_1^2\tau_2^2 \nonumber \\
&& \ \ + 2dn(n+2)m_2\tau_1^4 + 4n(n+2)(n+3)\tau_1^2\tau_3^2 \nonumber \\
&& \ \ + n(n+1)(n+2)(n+3)\tau_2^4.\label{prop1m}
\end{eqnarray}
\end{prop}

\begin{proof}
Formulas (\ref{prop1a}) and (\ref{prop1f}) are trivial
(notice that $\bb^TW\bb \sim \tau_1^2\chi^2_n$).  Formulas
(\ref{prop1b})-(\ref{prop1c}) may be found in \citep{letac2004all}.

Now let $\u_1,...,\u_d \in \R^d$ be an orthonormal basis of $\R^d$, with
$\bb = ||\bb||\u_1$.  Define the $d\times d$ symmetric matrices $H_{ij} =
(\u_i\u_j^T + \u_j\u_i^T)/2$ and $H_j = H_{1j}$, $i,j = 1,...,d$.  Then
\begin{equation}\label{prop1pfa}
\bb^TW^2\bb = \tau \sum_{j = 1}^d \tr(WH_j)^2.
\end{equation}
Since $S_2 = \{(1 \ 2), (1)(2)\}$,  the formula (\ref{letac}) and
Lemma 1 below imply
\begin{eqnarray*}
 E\tr(WH_j)^2 & = & 2^{2 - m((1 \ 2))}n^{m((1 \ 2))} \tr(\S H_j\S H_j)
 + 2^{2 - m((1)(2))}n^{m((1)(2))} \tr(\S H_j)^2 \\
& = & n\left\{(\u_1^T\S\u_j)^2 + \u_1^T\S\u_1\u_j^T\S\u_j\right\} + n^2(\u_1^T\S\u_j)^2.
\end{eqnarray*}
To prove (\ref{prop1h}), observe that 
\begin{eqnarray*}
E\bb^TW^2\bb & = & \tau_0^2 \sum_{j = 1}^d E\tr(WH_j)^2 \\
& = & n(n+1)\sum_{j= 1}^d \tau_0^2 (\u_1^T\S\u_j)^2 + n\sum_{j= 1}^d \tau_0^2
\u_1^T\S\u_1\u_j^T\S\u_j \\
& = & n(n+1)\tau_2^2 + dnm_1\tau_1^2.
\end{eqnarray*}
For (\ref{prop1i}), equation (\ref{letac}) implies 
\begin{eqnarray*}
E\left\{\tr(W)\bb^TW\bb\right\} & =& \tau_0^2
E\left\{\tr(W)\tr(WH_1)\right\} \\
& = & 2n\tau_0^2\tr(\S^2H_1) + n^2\tau_0^2\tr(\S)\tr(\S H_1) \\
& = & 2n\tau_2^2 + dn^2m_1\tau_1^2.
\end{eqnarray*}
To prove (\ref{prop1j}), first notice that
\begin{equation}\label{prop1pfb}
E\left\{\tr(W)\bb^TW^2\bb\right\} = \tau_0^2\sum_{j = 1}^d E\left\{\tr(W)\tr(W H_j)^2\right\}
\end{equation}
and that (\ref{letac}) implies
\[
 E\left\{\tr(W)\tr(W H_j)^2\right\} = \sum_{ \pi \in S_3}
2^{3-m(\pi)}n^{m(\pi)}r_{\pi}(\S)(I,H_j,H_j).
\]
It is clear that
\[\begin{array}{c}
r_{(1 \ 2 \ 3)}(\S)(I,H_j,H_j)  =  r_{(1 \ 3 \ 2)}(\S)(I,H_j,H_j)
\\
r_{(1 \ 2)(3)}(\S)(I,H_j,H_j) = r_{(1 \ 3)(2)}(\S)(I,H_j,H_j).
\end{array}
\]
Thus, by Lemma 1,
\begin{eqnarray*}
 E\left\{\tr(W)\tr(W H_j)^2\right\} & = & 8nr_{(1 \ 2 \
  3)}(\S)(I,H_j,H_j)+ 4n^2r_{(1 \ 2)(3)}(\S)(I,H_j,H_j) \\
&& \ \ + 2n^2r_{(1)(2 \ 3)}(\S)(I,H_j,H_j) + n^3r_{(1)(2)(3)}(\S)(I,H_j,H_j) \\
& = & 8n \tr(\S^2H_j\S H_j) + 4n^2\tr(\S^2H_j)\tr(\S^2H_j) \\
&& \ \ +
2n^2\tr(\S)\tr(\S H_j\S H_j) + n^3\tr(\S)\tr(\S H_j)^2 \\
& = & 2n(\u_1^T\S^2\u_1\u_j^T\S\u_j + \u_1^T\S\u_1\u_j^T\S^2\u_j + 2\u_1^T\S^2\u_j\u_1^T\S\u_j) \\
&& \ \ + 4n^2\u_1^T\S^2\u_j\u_1^T\S\u_j + n^2\tr(\S)\left\{(\u_1^T\S\u_j)^2 +
  \u_1^T\S\u_1\u_j^T\S\u_j\right\} \\
&& \ \ + n^3 \tr(\S)(\u_1^T\S\u_j)^2
\end{eqnarray*}
Combining this with (\ref{prop1pfb}) yields
\begin{eqnarray*}
E\left\{\tr(W)\bb^TW^2\bb\right\} & = & 2n \tau_0^2 \sum_{j=1}^d
\u_1^T\S^2\u_1\u_j^T\S\u_j  + 2n\tau_0^2\sum_{j = 1}^d \u_1^T\S\u_1\u_j^T\S^2\u_j  \\
&& \ \ + 4n(n+1) \tau_0^2\sum_{j = 1}^d \u_1^T\S^2\u_j\u_1^T\S\u_j + n^2\tr(\S)\tau_0^2\sum_{j = 1}^d \u_1^T\S\u_1\u_j^T\S\u_j \\
&& \ \ + n^2(n+1)\tr(\S)\tau_0^2\sum_{j= 1}^d (\u_1^T\S\u_j)^2 \\
& = & dn(n^2 + n + 2)m_1\tau_2^2 + 2dnm_2\tau_1^2 + 4d^2n(n+1)\tau_3^2 + d^2n^2m_1^2\tau_1^2.
\end{eqnarray*}
The proof of (\ref{prop1k}) is similar to the proof of
(\ref{prop1j}).  By (\ref{letac}) and Lemma 1,
\begin{eqnarray*}
E\left\{\tr(WH_1)\tr(WH_j)^2\right\} \!\! & =& \!\! 8nr_{(1 \ 2 \
  3)}(\S)(H_1,H_j,H_j) + 4n^2r_{(1 \ 2)(3)}(\S)(H_1,H_j,H_j) \\
&& \!\! \ \ + 2n^2r_{(1)(2 \ 3)}(\S)(H_1,H_j,H_j) + n^3r_{(1)(2)(3)}(\S)(H_1,H_j,H_j) \\
& = &\!\! 8n\tr(\S H_1\S H_j \S H_j) + 4n^2\tr(\S H_1 \S H_j)\tr(\S H_j) \\
&& \!\! \ \ + 2n^2\tr(\S H_1)\tr(\S H_j \S H_j) + n^3 \tr(\S H_1)\tr(\S H_j)^2 \\
& = &\!\! 2n\left\{(\u_1^T\S\u_1)^2\u_j^T\S\u_j  +
  3\u_1^T\S\u_1(\u_1^T\S\u_j)^2\right\}  + 4n^2\u_1^T\S\u_1(\u_1^T\S\u_j)^2 \\
&& \!\! \ \ + n^2\left\{(\u_1^T\S\u_1)^2\u_j^T\S\u_j +
  \u_1^T\S\u_1(\u_1^T\S\u_j)^2\right\} + n^3\u_1^T\S\u_1(\u_1^T\S\u_j)^2 \\
& = &\!\! n(n+2)(\u_1^T\S\u_1)^2\u_j^T\S\u_j  + n(n^2 + 5n + 6)\u_1^T\S\u_1(\u_1^T\S\u_j)^2.
\end{eqnarray*}
It follows that
\begin{eqnarray*}
E(\bb^TW\bb\bb^TW^2\bb) \!\! & = & \!\! \tau_0^4\sum_{j = 1}^d
\tr(WH_1)\tr(WH_j)^2  \\ & = & \!\! n(n+2)\sum_{j = 1}^d
\tau_0^4 (\u_1^T\S\u_1)^2\u_j^T\S\u_j + n(n^2 + 5n + 6)\sum_{j= 1}^d \tau_0^4
\u_1^T\S\u_1(\u_1^T\S\u_j)^2 \\
& = & \!\! dn(n+2)m_1\tau_1^4 + n(n^2 + 5n + 6)\tau_1^2\tau_2^2.
\end{eqnarray*}
To prove (\ref{prop1l}), consider the decomposition
\[
\bb^TW^3\bb = \tau_0^2\sum_{i,j = 1}^d \tr(WH_i)\tr(WH_j)\tr(WH_{ij}).  
\]
Equation(\ref{letac}) implies that
\[
E\left\{\tr(WH_i)\tr(WH_j)\tr(WH_{ij})\right\} = \sum_{\pi \in S_3}
2^{3 - m(\pi)}n^{m(\pi)} r_{\pi}(\S)(H_i,H_j,H_{ij}).  
\]
Since 
\begin{eqnarray*}
\sum_{i,j = 1}^d r_{(1 \ 2 \ 3)}(\S)(H_i,H_j,H_{ij}) & = & \sum_{i,j =
  1}^d r_{(1 \ 3 \ 2)}(\S)(H_i,H_j,H_{ij}) \\
\sum_{i,j = 1}^d r_{(1)(2 \ 3)}(\S)(H_i,H_j,H_{ij}) & = & \sum_{i,j =
  1}^d r_{(1 \ 3)(2)}(\S)(H_i,H_j,H_{ij}) ,
\end{eqnarray*}
it follows that
\begin{eqnarray}
E\bb^TW^3\bb \!\!\! &  = & \!\!\! 8n\tau_0^2 \sum_{i,j = 1}^d r_{(1 \ 2 \
  3)}(\S)(H_i,H_j,H_{ij}) + 4n^2\tau_0^2\sum_{i,j = 1}^d r_{(1)(2 \
  3)}(\S)(H_i,H_j,H_{ij})  \nonumber \\
&& \!\!\! \ + 2n^2\tau_0^2\sum_{i,j = 1}^d r_{(1 \
  2)(3)}(\S)(H_i,H_j,H_{ij}) + 
n^3\tau_0^2\sum_{i,j = 1}^d
r_{(1)(2)(3)}(\S)(H_i,H_j,H_{ij}). \label{prop1pfc}
\end{eqnarray}
By Lemma 1, 
\begin{eqnarray*}
\tau_0^2\sum_{i,j = 1}^d r_{(1 \ 2 \
  3)}(\S)(H_i,H_j,H_{ij})  & = & \tau_0^2 \sum_{i,j = 1}^d \tr(\S H_i\S H_j \S
H_{ij}) \\ & = & \frac{\tau_0^2}{8}\sum_{i,j = 1}^d 
\left\{\u_1^T\S\u_1\u_i^T\S\u_i\u_j^T\S\u_j  +
  \u_1^T\S\u_1(\u_i^T\S\u_j)^2 \right.\\
&& \ \ + (\u_1^T\S\u_i)^2\u_j^T\S\u_j + (\u_1^T\S\u_j)^2\u_i^T\S\u_i \\ && \ \ \left. + 4\u_1^T\S\u_i
  \u_1^T\S\u_j\u_i^T\S\u_j \right\} \\
& = & \frac{1}{8}\left(d^2m_1^2\tau_1^2 + dm_2\tau_1^2  + 2dm_1\tau_2^2 +
  4\tau_3^2\right) \\
\tau_0^2 \sum_{i,j = 1}^d r_{(1)(2 \
  3)}(\S)(H_i,H_j,H_{ij}) & = & \tau_0^2 \sum_{i,j = 1}^d \tr(\S H_i)\tr(\S
H_j \S H_{ij}) \\
& = & \frac{\tau_0^2}{2}\sum_{i,j = 1}^d
\u_1^T\S\u_i\left(\u_1^T\S\u_j\u_i^T\S\u_j + \u_1^T\S\u_i\u_j^T\S\u_j\right) \\
& =& \frac{1}{2}\left(\tau_3^2 + dm_1\tau_2^2\right) \\
\tau_0^2 \sum_{i,j = 1}^d r_{(1\ 2) (
  3)}(\S)(H_i,H_j,H_{ij}) & = &\tau_0^2 \sum_{i,j = 1}^d \tr(\S H_i \S
H_j)\tr(\S H_{ij}) \\
& = & \frac{\tau_0^2}{2}\sum_{i,j = 1}^d \left(\u_1^T\S\u_i \u_1^T\S \u_j
  + \u_1^T\S\u_1\u_i^T\S\u_j\right) \u_i^T\S\u_j \\
& =& \frac{1}{2}\left(\tau_3^2 + dm_2\tau_1^2\right) \\
\tau_0^2 \sum_{i,j = 1}^d r_{(1)(2)(3)}(\S)(H_i,H_j,H_{ij}) & =&
\tau_0^2\sum_{i,j = 1}^d \tr(\S H_i)\tr(\S H_j) \tr(\S H_{ij}) \\
& = & \tau_0^2 \sum_{i,j = 1}^d \u_1^T\S\u_i\u_1^T\S\u_j\u_i^T\S\u_j \\
& = & \tau_3^2.
\end{eqnarray*}
Using these results with (\ref{prop1pfc}) we obtain
\begin{eqnarray*}
E\bb^TW^3\bb & = & n\left(d^2m_1^2\tau_1^2 + dm_2\tau_1^2 + 2dm_1\tau_2^2 +
  4\tau_3^2\right) + 2n^2\left(\tau_3^2 + dm_1\tau_2^2\right)\\
&& \ \  +
n^2\left(\tau_3^2 + dm_2\tau_1^2\right) + n^3\tau_3^2 \\
& = & d^2nm_1^2\tau_1^2 + 2dn(n + 1)m_1\tau_2^2 + dn(n+1) m_2\tau_1^2 + (n^3 +
3n^2 + 4n)\tau_3^2.
\end{eqnarray*}
Finally, we prove (\ref{prop1m}).  Similar to the proof of
(\ref{prop1k})-(\ref{prop1l}), we have the decomposition
\[
(\bb^TW^2\bb)^2 = \tau_0^4 \sum_{i,j = 1}^d \tr(W H_i)^2\tr(WH_j)^2.
\]
By (\ref{letac}),
\[
E\left\{\tr(W H_i)^2\tr(WH_j)^2\right\} = \sum_{\pi \in S_4} 2^{4 - m(\pi)}n^{m(\pi)}r_{\pi}(\S)(H_i,H_i,H_j,H_j).
\]
It follows that
\[
E(\bb^TW^2\bb)^2 = \sum_{\pi \in S_4}2^{4 - m(\pi)}n^{m(\pi)}\tilde{r}_{\pi},
\]
where
\[
\tilde{r}_{\pi} = \sum_{i,j = 1}^d \tau_0^4r_{\pi}(\S)(H_i,H_i,H_j,H_j).
\]
One can easily see that
\begin{eqnarray*}
\tilde{r}_{(1 \ 2 \ 3 \ 4)} & = & \tilde{r}_{(1 \ 2 \ 4 \ 3)} \ \  = \
\ \tilde{r}_{(1  \ 3\ 4 \ 2)} \ \ = \ \ \tilde{r}_{(1 \ 4 \ 3 \ 2)} \\
\tilde{r}_{(1 \ 3 \ 2 \ 4)} & =&  \tilde{r}_{(1 \ 4 \ 2 \ 3)} \\
\tilde{r}_{(1)(2 \ 3 \ 4)}& = & \tilde{r}_{(1)(2 \ 4 \ 3)} \ \ = \ \
\tilde{r}_{(1 \ 3 \ 4)(2)}\ \  = \ \ \tilde{r}_{(1 \
  4 \ 3)(2)} \ \ = \ \ \tilde{r}_{(1 \ 2 \ 3)(4)} \\
& = & \tilde{r}_{(1 \ 3 \ 2)(4)} 
\ \ = \ \ \tilde{r}_{(1 \ 2 \ 4)(3)} \ \ = \ \ \tilde{r}_{(1 \ 4 \
  2)(3)}\\
\tilde{r}_{(1 \ 3)(2 \ 4)} & = & \tilde{r}_{(1 \ 4)(2 \ 3)} \\
\tilde{r}_{(1\ 2)(3)(4)} & = & \tilde{r}_{(1)(2)(3 \ 4)} \\
\tilde{r}_{(1 \ 3)(2)(4)} & = & \tilde{r}_{(1 \ 4)(2)(3)} \ \ = \ \
\tilde{r}_{(1)(3)(2 \ 4)} \ \ = \ \ \tilde{r}_{(1)(4)(2 \ 3)}. 
\end{eqnarray*}
Thus,
\begin{eqnarray}
E(\bb^TW^2\bb)^2 & = & 32n\tilde{r}_{(1 \ 2 \ 3 \ 4)} +
16n\tilde{r}_{(1 \ 3 \ 2 \ 4)} + 32n^2\tilde{r}_{(1)(2 \ 3 \ 4)}
+
8n^2\tilde{r}_{(1 \ 3)(2 \ 4)}\nonumber \\
&& \ \ + 4n^2\tilde{r}_{(1 \ 2)(3 \ 4)} +
4n^3\tilde{r}_{(1 \ 2)(3)(4)} + 8n^3\tilde{r}_{(1 \ 3)(2)(4)} + n^4\tilde{r}_{(1)(2)(3)(4)}. \label{prop1pfd}
\end{eqnarray}
It only remains to evaluate the $\tilde{r}_{\pi}$.  It follows from
Lemma 1 that
\begin{eqnarray*}
\tilde{r}_{(1 \ 2 \ 3 \ 4)} & = & \sum_{i,j = 1}^d \tau_0^4 \tr(\S H_i\S
H_i \S H_j \S H_j)\\
& = &   \sum_{i,j = 1}^d\frac{\tau_0^4}{16}\left\{2(\u_1^T\S\u_i)^2(\u_1^T\S\u_j)^2 +
  3\u_1^T\S\u_1(\u_1^T\S\u_i)^2\u_j^T\S\u_j  \right. \\
&& \ \ + 6\u_1^T\S\u_1\u_1^T\S\u_i\u_1^T\S\u_j\u_i^T\S\u_j +
3\u_1^T\S\u_1(\u_1^T\S\u_j)^2\u_i^T\S\u_i \\
&& \ \ \left. + (\u_1^T\S\u_1)^2\u_i^T\S\u_i\u_j\S\u_j  +
  (\u_1^T\S\u_1)^2(\u_i^T\S\u_j)^2\right\} \\
& = & \frac{1}{16}\left(2\tau_2^4 + 6dm_1\tau_1^2\tau_2^2 +
6\tau_1^2\tau_3^2  + d^2m_1^2\tau_1^4 + dm_2\tau_1^4\right)\\
\tilde{r}_{(1 \ 3 \ 2 \ 4)} & = & \sum_{i,j = 1}^d \tau_0^4 \tr(\S H_i\S
H_j \S H_i \S H_j)   \\
& = & \sum_{i,j =
  1}^d\frac{\tau_0^4}{8}\left\{(\u_1^T\S\u_i)^2(\u_1^T\S\u_j)^2
\right. \\
&& \ \ \left. +
  6\u_1^T\S\u_1\u_1^T\S\u_i\u_1^T\S\u_j\u_i^T\S\u_j +
  (\u_1^T\S\u_1)^2(\u_i^T\S\u_j)^2\right\} \\
& = & \frac{1}{8}(\tau_2^4 + 6\tau_1^2\tau_3^2 + dm_2\tau_1^4) \\
\tilde{r}_{(1)(2 \ 3 \ 4)}& = &  \sum_{i,j = 1}^d \tau_0^4\tr(\S
H_i)\tr(\S H_i \S H_j \S H_j) \\
& = & \sum_{i,j = 1}^d \frac{\tau_0^4}{4}
\left\{(\u_1^T\S\u_i)^2(\u_1^T\S\u_j)^2 +
  \u_1^T\S\u_1(\u_1^T\S\u_i)^2\u_j^T\S\u_j \right. \\
&& \ \ \left. +
  2\u_1^T\S\u_1\u_1^T\S\u_i\u_1^T\S\u_j\u_i^T\S\u_j\right\} \\
& = & \frac{1}{4}(\tau_2^4 + dm_1\tau_1^2\tau_2^2 + 2\tau_1^2\tau_3^2) \\
\tilde{r}_{(1 \ 3)(2 \ 4)} & = &  \sum_{i,j = 1}^d \tau_0^4\tr(\S
H_i \S H_j)^2\\
& = & \sum_{i,j=1}^d \frac{\tau_0^4}{4}\left\{\u_1^T\S\u_i\u_1^T\S\u_j +
  \u_1^T\S\u_1\u_i^T\S\u_j\right\}^2 \\
& = & \frac{1}{4}(\tau_2^4 + dm_2\tau_1^4 + 2\tau_1^2\tau_3^2) \\
\tilde{r}_{(1 \ 2)(3 \ 4)} & = &  \sum_{i,j = 1}^d \tau_0^4\tr(\S
H_i \S H_i)\tr(\S H_j \S H_j) \\
& = & \sum_{i,j = 1}^d \frac{\tau_0^4}{4}\left\{(\u_1^T\S\u_i)^2 +
  \u_1^T\S\u_1\u_i^T\S\u_i\right\} \left\{(\u_1^T\S\u_j)^2 +
  \u_1^T\S\u_1\u_j^T\S\u_j\right\} \\
& = & \frac{1}{4}(\tau_2^4 + 2dm_1\tau_1^2\tau_2^2 + d^2m_1^2\tau_1^4)\\
\tilde{r}_{(1 \ 2)(3)(4)} &= & \sum_{i,j = 1}^d \tau_0^4\tr(\S
H_i \S H_i)\tr(\S H_j)^2\\
& = & \sum_{i,j = 1}^d \frac{\tau_0^4 }{2}\left\{(\u_1^T\S\u_i)^2 +
  \u_1^T\S\u_1\u_i^T\S\u_i\right\}(\u_1^T\S\u_j)^2 \\
& = & \frac{1}{2}(\tau_2^4 + dm_1\tau_1^2\tau_2^2)\\
\tilde{r}_{(1 \ 3)(2)(4)} & = &  \sum_{i,j = 1}^d \tau_0^4\tr(\S
H_i \S H_j)\tr(\S H_i )\tr(\S H_j) \\
& = & \sum_{i,j = 1}^d
\frac{\tau_0^4}{2}\left\{\u_1^T\S\u_i\u_1^T\S\u_j +
  \u_1^T\S\u_1\u_i^T\S\u_j\right\}\u_1^T\S\u_i\u_1^T\S\u_j \\
& = & \frac{1}{2}(\tau_2^4 + \tau_1^2\tau_3^2) \\
\tilde{r}_{(1)(2)(3)(4)} & = & \sum_{i,j = 1}^d \tau_0^4\tr(\S
H_i)^2 \tr(\S H_j)^2 \\
& = & \sum_{i,j = 1}^d \tau_0^4 (\u_1^T\S\u_i)^2(\u_1^T\S\u_j)^2 \\
& = & \tau_2^4.
\end{eqnarray*}
Combining this with (\ref{prop1pfd}), we conclude that
\begin{eqnarray*}
E(\bb^TW^2\bb)^2 \!\!\! & = & \!\!\! 32n\tilde{r}_{(1 \ 2 \ 3 \ 4)} +
16n\tilde{r}_{(1 \ 3 \ 2 \ 4)} + 32n^2\tilde{r}_{(1)(2 \ 3 \ 4)}+
8n^2\tilde{r}_{(1 \ 3)(2 \ 4)} + 4n^2\tilde{r}_{(1 \ 2)(3 \ 4)} \\
&& \!\!\! \ \ +
4n^3\tilde{r}_{(1 \ 2)(3)(4)} + 8n^3\tilde{r}_{(1 \ 3)(2)(4)} + n^4\tilde{r}_{(1)(2)(3)(4)}
\\
& = & \!\!\! 2n(2\tau_2^4 + 6dm_1\tau_1^2\tau_2^2 + 6\tau_1^2\tau_3^2 + d^2m_1^2\tau_1^4
+ dm_2\tau_1^4) + 2n(\tau_2^4 + 6\tau_1^2\tau_3^2 + dm_2\tau_1^4) \\
&& \!\!\!\ \ + 8n^2(\tau_2^4 +
dm_1\tau_1^2\tau_2^2 + 2\tau_1^2\tau_3^2) + 2n^2(\tau_2^4 + dm_2\tau_1^4 + 2\tau_1^2\tau_3^2) \\
&& \!\!\!\ \ + n^2(\tau_2^4 +
2dm_1\tau_1^2\tau_2^2 + d^2m_1^2\tau_1^4) + 2n^3(\tau_2^4 + dm_1\tau_1^2\tau_2^2) + 4n^3(\tau_2^4 + \tau_1^2\tau_3^2)
+ n^4\tau_2^4 \\
& = & \!\!\! (n^4 + 6n^3 + 11n^2 + 6n)\tau_2^4 + d(2n^3 + 10n^2 +
12n)m_1\tau_1^2\tau_2^2 \\
&& \!\!\!\ \ + (4n^3 + 20n^2 + 24n)\tau_1^2\tau_3^2 + d^2(n^2 + 2n)m_1^2\tau_1^4 +
d(2n^2 + 4n)m_2\tau_1^4 \\
& = & \!\!\! d^2n(n+2)m_1^2\tau_1^4 + 2dn(n+2)(n+3)m_1\tau_1^2\tau_2^2 + 2dn(n+2)m_2\tau_1^4 \\
&& \!\!\! \ \ + 4n(n+2)(n+3)\tau_1^2\tau_3^2 + n(n+1)(n+2)(n+3)\tau_2^4.
\end{eqnarray*}
\end{proof}

\setcounter{lemma}{0}
    \renewcommand{\thelemma}{S\arabic{lemma}}

\begin{lemma}
Let $\u_1,...,\u_d \in \R^d$ and
define $H_j = (\u_1\u_j^T + \u_j\u_1^T)/2$. For integers $1 \leq i, j
\leq d$, we have
\begin{eqnarray}
\tr(\S H_{ij}) & = & \u_i^T\S\u_j \label{lemma1a} \\
\tr(\S H_i \S H_j) & =  & \frac{1}{2}\left(\u_1^T\S\u_i\u_1^T\S\u_j +
  \u_1^T\S\u_1\u_i^T\S\u_j\right) \label{lemma1b} \\
\tr(\S H_i\S H_{ij}) & = & \frac{1}{2}\left(\u_1^T\S\u_i\u_i\S\u_j +
  \u_1^T\S\u_j\u_i^T\S\u_i\right) \label{lemma1c} \\
\tr(\S^2H_i\S H_j) & = & \frac{1}{4}\left(\u_1^T\S^2\u_1\u_i^T\S\u_j+
  \u_1^T\S\u_1\u_i^T\S^2\u_j \right. \nonumber \\ && \ \
\left.  + \u_1^T\S^2\u_i\u_1^T\S\u_j
  +\u_1^T\S\u_i\u_1^T\S^2\u_j\right\} \label{lemma1d} \\
\tr(\S H_i \S H_j \S H_j) & = & \frac{1}{4}\left\{\u_1^T\S\u_i(\u_1^T\S\u_j)^2 +
  \u_1^T\S\u_1\u_1^T\S\u_i\u_j^T\S\u_j\right. \nonumber \\
&& \ \ \left. +
  2\u_1^T\S\u_1\u_1^T\S\u_j\u_i^T\S\u_j\right\} \label{lemma1e} \\
\tr(\S H_i \S H_j \S H_{ij}) & = &
\frac{1}{8}\left\{\u_1^T\S\u_1\u_i^T\S\u_i\u_j^T\S\u_j +
  \u_1^T\S\u_1(\u_i^T\S\u_j)^2 \right. \nonumber \\ 
&& \ \ + (\u_1^T\S\u_i)^2\u_j^T\S\u_j + (\u_1^T\S\u_j)^2\u_i^T\S\u_i
\nonumber \\
&& \ \ \left. + 4\u_1^T\S\u_i \u_1^T\S\u_j\u_i^T\S\u_j \right\}\label{lemma1f}\\
\tr(\S H_i \S H_i \S H_j \S H_j) & = &
\frac{1}{16}\left\{2(\u_1^T\S\u_i)^2(\u_1^T\S\u_j)^2
  3\u_1^T\S\u_1(\u_1^T\S\u_i)^2\u_j^T\S\u_j  \right. \nonumber \\
&& \ \ + 6\u_1^T\S\u_1\u_1^T\S\u_i\u_1^T\S\u_j\u_i^T\S\u_j +
3\u_1^T\S\u_1(\u_1^T\S\u_j)^2\u_i^T\S\u_i \nonumber \\
&& \ \ \left. + (\u_1^T\S\u_1)^2\u_i^T\S\u_i\u_j\S\u_j +
  (\u_1^T\S\u_1)^2(\u_i^T\S\u_j)^2\right\} \label{lemma1g} \\ \nonumber
\tr(\S H_i \S H_j \S H_i \S H_j) & = &
\frac{1}{8}\left\{(\u_1^T\S\u_i)^2(\u_1^T\S\u_j)^2 +
  6\u_1^T\S\u_1\u_1^T\S\u_i\u_1^T\S\u_j\u_i^T\S\u_j \right. \\ && \ \
\left. + (\u_1^T\S\u_1)^2(\u_i^T\S\u_j)^2\right\} \label{lemma1h}
\end{eqnarray}
\end{lemma}

\begin{proof}
The identity (\ref{lemma1a}) is trivial.  To prove (\ref{lemma1b}), we
have
\begin{eqnarray*}
\tr(\S H_i \S H_j) & = & \frac{1}{4}\tr\left\{\S(\u_1\u_i^T +
  \u_i\u_1^T)\S(\u_1\u_j^T + \u_j\u_1^T)\right\} \\
& = & \frac{1}{4}\tr\left(\S\u_1\u_i^T\S\u_1\u_j^T +
  \S\u_1\u_i^T\S\u_j\u_1^T + \S\u_i\u_1^T\S\u_1\u_j^T +
  \S\u_i\u_1^T\S\u_j\u_1^T\right) \\
& = & \frac{1}{2}\left(\u_1^T\S\u_i\u_1^T\S\u_j +
  \u_1^T\S\u_1\u_i^T\S\u_j\right).
\end{eqnarray*}
Equation (\ref{lemma1c}) follows from
\begin{eqnarray*}
\tr(\S H_i \S H_{ij}) & = & \frac{1}{4}\tr\left\{\S(\u_1\u_i^T +
  \u_i\u_1^T)\S(\u_i\u_j^T + \u_j\u_i^T)\right\} \\
& = & \frac{1}{4}\tr\left(\S\u_1\u_i^T\S\u_i\u_j^T +
  \S\u_1\u_i^T\S\u_j\u_i^T  + \S\u_i\u_1^T\S\u_i\u_j^T +
  \S\u_i\u_1^T\S\u_j\u_i^T\right) \\
& = & \frac{1}{2}\left(\u_1^T\S\u_i\u_i^T\S\u_j + \u_1^T\S\u_j\u_i^T\S\u_i\right). 
\end{eqnarray*}
For (\ref{lemma1d}), we have
\begin{eqnarray*}
\tr(\S^2 H_i \S H_j) & = & \frac{1}{4}\tr\left\{\S^2(\u_1\u_i^T +
  \u_i\u_1^T)\S(\u_1\u_j^T + \u_j\u_1^T)\right\} \\
& = & \frac{1}{4}\tr\left(\S^2\u_1\u_i^T\S\u_1\u_j^T +
  \S^2\u_1\u_i^T\S\u_j\u_1^T  + \S^2\u_i\u_1^T\S\u_1\u_j^T +
  \S^2\u_i\u_1^T\S\u_j\u_1^T\right) \\
& = & \frac{1}{4}\left(\u_1^T\S\u_i\u_1^T\S^2\u_j +
  \u_1^T\S^2\u_1\u_i^T\S\u_j  +
  \u_1^T\S\u_1\u_i^T\S^2\u_j + \u_1^T\S^2\u_i \u_1\S\u_j\right).  
\end{eqnarray*}
To prove (\ref{lemma1e})-(\ref{lemma1f}), observe that 
\begin{eqnarray*}
\tr(\S H_i \S H_j \S H_j) & = & \frac{1}{8}\tr\left\{\S(\u_1\u_i^T +
  \u_i\u_1^T)\S(\u_1\u_j^T + \u_j\u_1^T) \S(\u_1\u_j^T + \u_j\u_1^T)\right\} \\
& = & \frac{1}{8}\tr\left(\S\u_1\u_i^T\S\u_1\u_j^T\S\u_1\u_j^T +
  \S\u_1\u_i^T\S\u_1\u_j^T\S\u_j\u_1^T \right. \\
&& \ \ + \S\u_1\u_i^T\S\u_j\u_1^T\S\u_1\u_j^T +
  \S\u_1\u_i^T\S\u_j\u_1^T\S\u_j\u_1^T \\
&& \ \ + \S\u_i\u_1^T\S\u_1\u_j^T\S\u_1\u_j^T +
  \S\u_i\u_1^T\S\u_1\u_j^T\S\u_j\u_1^T  \\
&& \ \ \left. + \S\u_i\u_1^T\S\u_j\u_1^T\S\u_1\u_j^T +
  \S\u_i\u_1^T\S\u_j\u_1^T\S\u_j\u_1^T\right) \\
& = & \frac{1}{4}\left\{\u_1^T\S\u_i(\u_1^T\S\u_j)^2 +
  \u_1^T\S\u_1\u_1^T\S\u_i\u_j^T\S\u_j  + 2\u_1^T\S\u_1\u_1^T\S\u_j\u_i^T\S\u_j\right\}
\end{eqnarray*}
and
\begin{eqnarray*}
\tr(\S H_i \S H_j \S H_{ij}) & = & \frac{1}{8}\tr\left\{\S(\u_1\u_i^T +
  \u_i\u_1^T)\S(\u_1\u_j^T + \u_j\u_1^T) \S(\u_i\u_j^T + \u_j\u_i^T)\right\} \\
& = & \frac{1}{8}\tr\left(\S\u_1\u_i^T\S\u_1\u_j^T\S\u_i\u_j^T +
  \S\u_1\u_i^T\S\u_1\u_j^T\S\u_j\u_i^T \right. \\
&& \ \ + \S\u_1\u_i^T\S\u_j\u_1^T\S\u_i\u_j^T +
  \S\u_1\u_i^T\S\u_j\u_1^T\S\u_j\u_i^T \\
&& \ \ + \S\u_i\u_1^T\S\u_1\u_j^T\S\u_i\u_j^T +
  \S\u_i\u_1^T\S\u_1\u_j^T\S\u_j\u_i^T \\
&& \ \ \left.+ \S\u_i\u_1^T\S\u_j\u_1^T\S\u_i\u_j^T +
  \S\u_i\u_1^T\S\u_j\u_1^T\S\u_j\u_i^T \right) \\
& = & \frac{1}{8}\left\{\u_1^T\S\u_1\u_i^T\S\u_i\u_j^T\S\u_j +
  \u_1^T\S\u_1(\u_i^T\S\u_j)^2+ (\u_1^T\S\u_i)^2\u_j^T\S\u_j  \right.  \\ 
&& \ \ \left. + (\u_1^T\S\u_j)^2\u_i^T\S\u_i
+ 4\u_1^T\S\u_i \u_1^T\S\u_j\u_i^T\S\u_j \right\}.
\end{eqnarray*}
Finally, to prove (\ref{lemma1g})-(\ref{lemma1h}), we have
\begin{eqnarray*}
&& \!\!\!\!\!\!\!\!\!\!\!\!\!\!\!\! \! \! \! \! \! \! \! \! \! \! \tr(\S H_i \S H_i \S H_j \S H_j) \\ && \ \ =  \frac{1}{16}\tr\left\{\S(\u_1\u_i^T +
  \u_i\u_1^T) \S(\u_1\u_i^T +
  \u_i\u_1^T) \S(\u_1\u_j^T + \u_j\u_1^T) \S(\u_1\u_j^T +
  \u_j\u_1^T)\right\} \\
&& \ \ = \frac{1}{16}\tr\left(
  \S\u_1\u_i^T\S\u_1\u_i^T\S\u_1\u_j^T\S\u_1\u_j^T +
  \S\u_1\u_i^T\S\u_1\u_i^T\S\u_1\u_j^T\S\u_j\u_1^T \right. \\
&& \qquad \ \  + \S\u_1\u_i^T\S\u_1\u_i^T\S\u_j\u_1^T\S\u_1\u_j^T +
  \S\u_1\u_i^T\S\u_1\u_i^T\S\u_j\u_1^T\S\u_j\u_1^T \\ 
&& \qquad \ \ + \S\u_1\u_i^T\S\u_i\u_1^T\S\u_1\u_j^T\S\u_1\u_j^T +
  \S\u_1\u_i^T\S\u_i\u_1^T\S\u_1\u_j^T\S\u_j\u_1^T \\
&& \qquad \ \  + \S\u_1\u_i^T\S\u_i\u_1^T\S\u_j\u_1^T\S\u_1\u_j^T +
  \S\u_1\u_i^T\S\u_i\u_1^T\S\u_j\u_1^T\S\u_j\u_1^T \\
&& \qquad \ \ +  \S\u_i\u_1^T\S\u_1\u_i^T\S\u_1\u_j^T\S\u_1\u_j^T +
  \S\u_i\u_1^T\S\u_1\u_i^T\S\u_1\u_j^T\S\u_j\u_1^T \\
&& \qquad \ \  + \S\u_i\u_1^T\S\u_1\u_i^T\S\u_j\u_1^T\S\u_1\u_j^T +
  \S\u_i\u_1^T\S\u_1\u_i^T\S\u_j\u_1^T\S\u_j\u_1^T \\ 
&& \qquad \ \ + \S\u_i\u_1^T\S\u_i\u_1^T\S\u_1\u_j^T\S\u_1\u_j^T +
  \S\u_i\u_1^T\S\u_i\u_1^T\S\u_1\u_j^T\S\u_j\u_1^T \\
&& \qquad \ \ \left. + \S\u_i\u_1^T\S\u_i\u_1^T\S\u_j\u_1^T\S\u_1\u_j^T +
  \S\u_i\u_1^T\S\u_i\u_1^T\S\u_j\u_1^T\S\u_j\u_1^T\right) \\
&&  \ \ = \frac{1}{16}\left\{2(\u_1^T\S\u_i)^2(\u_1^T\S\u_j)^2 +
  3\u_1^T\S\u_1(\u_1^T\S\u_i)^2\u_j^T\S\u_j  \right. \\
&& \qquad \ \ + 6\u_1^T\S\u_1\u_1^T\S\u_i\u_1^T\S\u_j\u_i^T\S\u_j +
3\u_1^T\S\u_1(\u_1^T\S\u_j)^2\u_i^T\S\u_i \\
&& \qquad \ \ \left.+ (\u_1^T\S\u_1)^2\u_i^T\S\u_i\u_j\S\u_j + (\u_1^T\S\u_1)^2(\u_i^T\S\u_j)^2\right\}
\end{eqnarray*}
and
\begin{eqnarray*}
&& \!\!\!\!\!\!\!\!\!\!\! \! \! \! \! \! \! \! \! \! \! \tr(\S H_i \S H_j \S H_i \S H_j) \\ && \ \ =  \frac{1}{16}\tr\left\{\S(\u_1\u_i^T +
  \u_i\u_1^T) \S(\u_1\u_j^T +
  \u_j\u_1^T) \S(\u_1\u_i^T + \u_i\u_1^T) \S(\u_1\u_j^T +
  \u_j\u_1^T)\right\} \\
&& \ \ =
\frac{1}{16}\tr\left(\S\u_1\u_i^T\S\u_1\u_j^T\S\u_1\u_i^T\S\u_1\u_j^T
  + \S\u_1\u_i^T\S\u_1\u_j^T\S\u_1\u_i^T\S\u_j\u_1^T \right. \\
&& \qquad \ \ + \S\u_1\u_i^T\S\u_1\u_j^T\S\u_i\u_1^T\S\u_1\u_j^T
  + \S\u_1\u_i^T\S\u_1\u_j^T\S\u_i\u_1^T\S\u_j\u_1^T \\
&& \qquad \ \ + \S\u_1\u_i^T\S\u_j\u_1^T\S\u_1\u_i^T\S\u_1\u_j^T
  + \S\u_1\u_i^T\S\u_j\u_1^T\S\u_1\u_i^T\S\u_j\u_1^T \\
&& \qquad \ \ + \S\u_1\u_i^T\S\u_j\u_1^T\S\u_i\u_1^T\S\u_1\u_j^T
  + \S\u_1\u_i^T\S\u_j\u_1^T\S\u_i\u_1^T\S\u_j\u_1^T \\
&& \qquad \ \ + \S\u_i\u_1^T\S\u_1\u_j^T\S\u_1\u_i^T\S\u_1\u_j^T
  + \S\u_i\u_1^T\S\u_1\u_j^T\S\u_1\u_i^T\S\u_j\u_1^T \\
&& \qquad \ \ + \S\u_i\u_1^T\S\u_1\u_j^T\S\u_i\u_1^T\S\u_1\u_j^T
  + \S\u_i\u_1^T\S\u_1\u_j^T\S\u_i\u_1^T\S\u_j\u_1^T \\
&& \qquad \ \ + \S\u_i\u_1^T\S\u_j\u_1^T\S\u_1\u_i^T\S\u_1\u_j^T
  + \S\u_i\u_1^T\S\u_j\u_1^T\S\u_1\u_i^T\S\u_j\u_1^T \\
&& \qquad \ \ \left.+ \S\u_i\u_1^T\S\u_j\u_1^T\S\u_i\u_1^T\S\u_1\u_j^T
  + \S\u_i\u_1^T\S\u_j\u_1^T\S\u_i\u_1^T\S\u_j\u_1^T \right) \\
&& \ \ = \frac{1}{8}\left\{(\u_1^T\S\u_i)^2(\u_1^T\S\u_j)^2 +
  6\u_1^T\S\u_1\u_1^T\S\u_i\u_1^T\S\u_j\u_i^T\S\u_j + (\u_1^T\S\u_1)^2(\u_i^T\S\u_j)^2\right\}
\end{eqnarray*}
\end{proof}

}
\end{document}